\documentclass[aps,pra,amsmath,notitlepage,nofootinbib,twocolumn,superscriptaddress,noeprint]{revtex4-2}
\usepackage{amsmath,amsfonts,amssymb,amsthm}
\usepackage{mathtools}
\usepackage[pdftex]{graphicx}
\usepackage[usenames, dvipsnames, svgnames, table]{xcolor}
\usepackage[colorlinks=true, citecolor=Blue, linkcolor=BrickRed, urlcolor=Brown]{hyperref}
\usepackage{txfonts}
\usepackage{mathrsfs}
\usepackage{bm}
\usepackage{multirow}
\usepackage{braket}
\usepackage{import}
\usepackage[ruled, noline, noend]{algorithm2e}
\usepackage{array}
\usepackage{comment}
\usepackage{tikz}
\usetikzlibrary{quantikz}
\usepackage{graphicx}
\usepackage[export]{adjustbox}

\newtheorem{theorem}{Theorem}
\newtheorem{lemma}{Lemma}
\newtheorem{Note}{Note}

\newtheorem{corollary}{Corollary}
\theoremstyle{definition}
\newtheorem{definition}{Definition}
\newcommand{\be}{\begin{equation}}
\newcommand{\ee}{\end{equation}}
\newcommand{\ben}{\begin{eqnarray}}
\newcommand{\een}{\end{eqnarray}}
\newcommand{\bes}{\begin{subequations}}
\newcommand{\ees}{\end{subequations}}
\newcommand{\bF}{\begin{figure}}
\newcommand{\eF}{\end{figure}}

\newcommand{\orcid}[1]{\href{https://orcid.org/#1}{\includegraphics[height = 2ex]{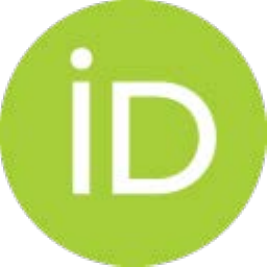}}}

\begin{document}

\title{Improved Accreditation of Analogue Quantum Simulation and Establishing Quantum Advantage}

\author{Andrew Jackson \orcid{0000-0002-5981-1604}}
\email{Andrew.J.Jackson@ed.ac.uk}
\affiliation{Department of Physics, University of Warwick, Coventry, CV4 7AL, United Kingdom}
\affiliation{School of Informatics, University of Edinburgh, Edinburgh, EH8 9AB, United Kingdom}

\author{Animesh Datta \orcid{0000-0003-4021-4655}}
\email{Animesh.Datta@warwick.ac.uk}
\affiliation{Department of Physics, University of Warwick, Coventry, CV4 7AL, United Kingdom}


\begin{abstract}
    We improve on the results of [A. Jackson et al. Proc. Natl. Acad. Sci. U.S.A 121 (6). 2024] on the verification of analogue quantum simulators by eliminating the use of universal Hamiltonians, removing the need for two-qubit gates, and no longer assuming error is represented by \emph{identical} maps across simulations.
    This new protocol better reflects the reality of extant analogue simulators.
    It integrates well with recent complexity theoretic results, leading to a near-term feasible simulation-based route to establishing quantum advantage.
\end{abstract}

\date{\today}
\maketitle

\section{Introduction}
Quantum computation and simulation are likely to be transformative for the study of fundamental physics as well as applications.
Key examples include the simulation of quantum systems~\cite{Lloyd, Poulin2015TheTS, Heyleaau8342, benedetti2020hardwareefficient, Clinton_2024} and the computation of partition functions~\cite{Chowdhury_2021, Jackson_2023, doi:10.1137/1.9781611977554.ch46}.
While the circuit model of quantum computing gets the bulk of the attention; the analogue model of quantum computing is capable of far greater feats of simulation than current digital computers, due to two important factors:
\begin{enumerate}
    \item The number of gates required for a modest (for an analogue simulator) simulation on a digital quantum computer quickly becomes prohibitive. 
    For instance, a simulation of the Hubbard model on a $10 \times 10$ lattice for a duration of $10 \hbar J^{-1}$ is already feasible on analogue simulators, however a digital quantum computer would require over a million gates~\cite{Daley2022}.
    \item Analogue simulators are typically far larger than their contemporaneous digital counterparts~\cite[Table.~1]{Daley2022}.
\end{enumerate}
It is therefore important to develop methods (known as verification protocols) to ascertain the correctness of noisy, error-prone real-world analogue quantum simulations. 
This is also important as verification is a key to demonstrating quantum advantage. The confidence necessary to establish the performing of the required task to demonstrate quantum advantage is provided by verification protocols. 

A key family of methods to ascertain the correctness of analogue quantum simulators is Hamiltonian learning~\cite{PhysRevLett.122.020504, PRXQuantum.2.010102}. 
By experimentally time-evolving a state that is invariant under the target time evolution, and estimating specific expected values of the resulting state, the experimentally applied Hamiltonian can be estimated and compared to the target Hamiltonian. However, this assimilates both state preparation and measurement errors into that of the Hamiltonian. Furthermore, it
does not translate into quantitative guarantees on the output of any given simulation.

Another approach to analogue verification consists of running a time evolution forward and then backward for the same duration: if there is no error, this returns the system to its initial state, which can be checked. This approach was furthered in Ref.~\cite{Shaffer2021}, which introduced a protocol that evolves an initial state through a closed loop; returning the system to the initial state. This provides a measure of how faithfully the simulator has implemented the requested Hamiltonian but, again, cannot produce quantitative guarantees on a 
given simulation's output.

More recently, analogue quantum accreditation~\cite{jackson2023accreditation} exploited universality~\cite{Cubitt_2018, zhou2021strongly, PRXQuantum.3.010308} to emulate any desired simulation using an XY-interaction Hamiltonian. It then used the anti-commutation properties of the XY-interaction to exactly invert\footnote{This means to reverse the direction of the time evolution while preserving the Hamiltonian governing the time evolution.} half of the desired simulation, which then cancelled out with the other half. Hence the overall simulation was reduced -- in the ideal error-free case -- to the identity.
In the actual, noisy and error-prone case, errors could be reduced to stochastic Pauli error~\cite{Wallman_2016, Hashim_2021} then detected. Notably, this protocol gives rigorous bounds on the ideal-actual variation distance (as in Def.~\ref{def:idealActualVarDist}) of the given simulation's outputs.
\begin{definition}
    \label{def:idealActualVarDist}
    For any execution, $\Tilde{\mathcal{C}}$, of a simulation, $\mathcal{C}$, its ideal-actual variation distance, $\nu \big( \Tilde{\mathcal{C}} \big)$, is the 
    total variation distance~\cite{feng2023deterministically} between the probability distribution it would sample from in the ideal case if there were no error and the distribution it actually samples from in the actual noisy, error-prone case.\\
    The ideal-actual variation distance is given by:
\begin{equation}
\label{eq:tvd}
\nu \big( \Tilde{\mathcal{C}} \big) = \dfrac{1}{2} \sum_{s \in \Omega} \bigg \vert P \big( s \big) - \Tilde{P}\big( s \big) \bigg \vert,
\end{equation}
where $\Tilde{P}(s)$ is the probability distribution of the measurement outcomes of the actual, erroneous simulation (i.e. $\Tilde{\mathcal{C}}$), which may differ from the ideal (i.e. error-free) probability distribution of the measurement outcomes, $P(s)$ (which is the distribution $\mathcal{C}$ is intended to sample from); and $\Omega$ is the sample space of the measurement outcomes.
\end{definition}

The accreditation protocol constructed here will provide an upper bound
$\epsilon_{\text{VD}} \in [0,1)$ such that $\nu \big( \Tilde{\mathcal{C}} \big) \leq \epsilon_{\text{VD}}.$
$\nu \big( \Tilde{\mathcal{C}} \big)$ is also central in some approaches to establishing quantum advantage~\cite{PhysRevLett.118.040502, PhysRevX.8.021010,
Kapourniotis_2019, ringbauer2024verifiablemeasurementbasedquantumrandom,liu2024efficientlyverifiablequantumadvantage}.

The key development in this paper is the move away from the reduction of Hamiltonians we want to simulate to XY-interaction Hamiltonians. 
This is enabled by aiming to only \textit{approximately} invert Hamiltonians, instead of the exact inversion used in Ref.~\cite{jackson2023accreditation} that required the reduction to an XY-interaction.
Using the results of Ref.~\cite{odake2024universal}, we can approximately invert any spin Hamiltonian to arbitrarily small error. This allows for an analogue verification protocol that does not require additional qubits, or mapping the Hamiltonian of interest to a universal Hamiltonian, or two-qubit gates. The trade-off is an increase in the number of single-qubit gates used and how much the target time evolution is split into many smaller time evolutions. These increases lead to greater error in the simulation.

Additionally, in Lemma~\ref{lem:twiceProbBound} (in Sec.~\ref{sec:Trap2Targ}), by taking a new approach to relating the error in trap and target simulations we remove the assumption that error is identical across distinct executions of the same circuit, as was required in Ref.~\cite{jackson2023accreditation}.

Finally, we discuss, in Sec.~\ref{sec:quantumAdvantage}, how the accreditation protocol presented herein -- which provides the ideal-actual variation distance (as  in Def.~\ref{def:idealActualVarDist}) of an analogue quantum simulation, can be readily used in any experiment aiming to demonstrate quantum advantage via analogue quantum simulation. These experiments are based on establishing  quantum advantage under complexity theoretic assumptions~\cite{PhysRevX.8.021010, ringbauer2024verifiablemeasurementbasedquantumrandom}. Their central result is that certain simulations $\mathcal{C}$, performed with $\nu \big( \Tilde{\mathcal{C}} \big)  \leq1 - 1/\sqrt{2} \approx 0.292$ are not classically tractable under certain assumptions.

The main result of this paper is Theorem~\ref{finalTheorem}. It shows that Protocol~\ref{fullProtocolSketch} (the protocol developed herein) functions as an accreditation protocol, as defined in Def.~\ref{AAQSDef}. It improves upon Ref.~\cite{jackson2023accreditation}; as such, the protocol presented therein is the most pertinent to compare our new protocol with. Therefore, we provide a comparison of the two protocols in Table~\ref{tab:protocolComparison}.
\begin{table}[h!]
\begin{tabular}{lcc}
                               & \underline{Ref.~\cite{jackson2023accreditation}} & \underline{This Paper}                                  \\
Uses Universal Hamiltonians & Yes                                           & No                                            \\
Potentially adds Extra Qubits              & Yes                                           & No                                            \\
Error is assumed to be ID                   & Yes                                           & No                                            \\
Requires Two-Qubit Gates       & Yes                                           & No                                            \\
Inversion is Approximate       & No                                            & Yes                                           \\
Number of Single-Qubit Gates   & $O (N)$                             & $O \big( N / \epsilon \big)$
\end{tabular}
\caption{Comparison of the protocol presented herein with that in Ref.~\cite{jackson2023accreditation}. $N$ is the number of qubits in the target Hamiltonian, $\epsilon$ is a chosen upper bound on the additive (or absolute) error in the diamond norm in implementing the approximately inverted time evolution, and ID stands for Identically Distributed.}
    \label{tab:protocolComparison}
\end{table}

The presentation of this new inversion method  (in Sec.~\ref{sec:invertingHamiltonians}) is the first section of this paper. 
In Sec. III, we define what we mean by an analogue quantum simulation and what it means to accredit an  analogue quantum simulation.
We then continue -- in Sec.~\ref{ErrorModelSubsection} -- with the presentation of the error model used throughout this paper. We then define our new analogue accreditation protocol and how it is used (in Sec.~\ref{sec:protocolPresent}), before concluding with a discussion of our results.

    Throughout this paper, operators (i.e. unitary operators and CPTP maps) are denoted by caligraphic Roman letters (e.g. $\mathcal{A}, \mathcal{B}, \mathcal{C}$) and sets are denoted by bubble Roman letters (e.g. $\mathbb{A}, \mathbb{B}, \mathbb{C}$); except when they define a graph, when they are bold (e.g. $\mathbf{A}, \mathbf{B}, \mathbf{C}$).
    Additionally, in all definitions, the concept being defined is underlined.

\section{Inversion of General Spin Hamiltonians}
\label{sec:invertingHamiltonians}
The principal source of the improvement this paper presents is a method of approximately inverting any Hamiltonian~\cite{odake2024universal}.
It is used in Theorem~\ref{thm:alwaysInvertible}. 

\begin{theorem}
    \label{thm:alwaysInvertible}
    For any $N$-qubit spin Hamiltonian, $\mathcal{H}$, and time duration, $t \in \mathbb{R}$, there exists a sub-circuit (that we denote $\mathcal{B} (\mathcal{H}, t, \epsilon, 0)$) implementing $e^{-i \mathcal{H}t}$ that, purely by adding single-qubit gates, can be transformed into a sub-circuit (that we denote $\mathcal{B} (\mathcal{H}, t, \epsilon, 1)$) which approximately implements $e^{i \mathcal{H}t / (L-1)}$ to arbitrary additive error in the diamond norm, $\epsilon \in \mathbb{R}$ , where $L \in \mathbb{N}$ is efficiently computable.
    Both $\mathcal{B} (\mathcal{H}, t, \epsilon, 0)$ and $\mathcal{B} (\mathcal{H}, t, \epsilon, 1)$ require splitting the time evolution into $ O(t^2/\epsilon)$ parts, and $\mathcal{B} (\mathcal{H}, t, \epsilon, 1)$ also uses $O(N t^2/\epsilon)$ single-qubit gates.
\end{theorem}

Theorem~\ref{thm:alwaysInvertible} is an original result that uses Ref.~\cite[Algorithm 5]{odake2024universal}.
Theorem~\ref{thm:alwaysInvertible} will be used in Sec.~\ref{sec:defTrapAndTarg}, when constructing trap and target simulations, and these simulations will only differ in the value of $j$ in their use of a single $\mathcal{B} (\mathcal{H}, t, \epsilon, j)$ sub-circuit. Therefore, due to Theorem~\ref{thm:alwaysInvertible}, they will only differ in their single-qubit gates but in the traps this will cause one half of the time evolution to cancel with the other (so the entire time evolution is equivalent to the identity).

The rest of Sec.~\ref{sec:invertingHamiltonians} is devoted to proving Theorem~\ref{thm:alwaysInvertible} culminating in its formal proof. 
We begin with the definition of $\mathcal{B} (\mathcal{H}, t, \epsilon, 0)$ mentioned in Theorem~\ref{thm:alwaysInvertible}.

\subsection{Defining the Forwards Evolution, $\mathcal{B} (\mathcal{H}, t, \epsilon, 0)$}
The construction of  $\mathcal{B} (\mathcal{H}, t, \epsilon, 0)$ is presented in Algorithm~\ref{alg:forwardTime} and is depicted\footnote{We use circuits as a presentational tool to show both analogue and hybrid simulations, but note that we definitively do not use the circuit model within this paper.} in Fig.~\ref{alg:forwardTime}.
\begin{figure}
    \centering
\begin{algorithm}[H]
$\mathbf{Input:}$ \\
$\bullet$ A Hamiltonian, $\mathcal{H}$, where $c_j$ is the coefficient of its $j$th term\\
$\bullet$ A evolution duration, $t \in \mathbb{R}$\\
$\bullet$ A maximum permissible additive error -- in terms of the diamond norm -- in the corresponding approximate time inversion, $\epsilon \in \mathbb{R}$
\vspace{0.3cm}
    \begin{enumerate}
        \item Calculate the $\mathbb{G}'_{\mathcal{H}}$ required to invert $\mathcal{H}$, as defined in Def.~\ref{def:G'}
        \item Calculate $L = \vert \mathbb{G}'_{\mathcal{H}} \vert$
        \item Calculate $M =  \dfrac{2 t^2  \big[ \sum_{j = 1} \big( \vert c_j \vert \big) \big]^2 L}{\epsilon (L - 1)} $
        \item Initialize a blank circuit, $\textit{circ} $
        \item Set $\text{timeStep} = 1$
        \item \While{\text{timeStep} $\leq M$}{
            \begin{enumerate}
                \item Add $\mathcal{I}^{\otimes N} e^{-i t \mathcal{H} / M}  \mathcal{I}^{\otimes N}$ to $\textit{circ}$
                \item $\text{timeStep } += 1$
            \end{enumerate}
        }
    \end{enumerate}
\vspace{0.1cm}
$\mathbf{Return}:$ $\textit{circ}$
\caption{$\mathcal{B}(\mathcal{H}, t, \epsilon, 0)$ construction algorithm
 \label{alg:forwardTime}}
\end{algorithm}
\end{figure}
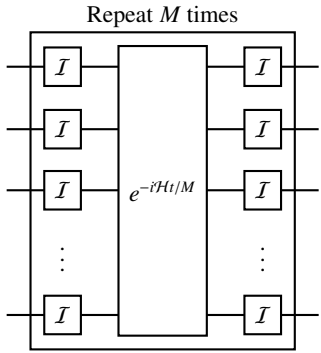
\begin{figure}[h]
    \centering
 \begin{quantikz}[row sep=0.3cm]
\qw & \gate{\mathcal{I}}  \gategroup[5,steps=3,style={inner
sep=2pt}]{Repeat $M$ times} & \gate[wires=5, nwires=4]{e^{-i \mathcal{H} t / M }} & \gate{\mathcal{I}} & \qw &\\
\qw & \gate{\mathcal{I}} && \gate{\mathcal{I}} & \qw &\\
\qw & \gate{\mathcal{I}} && \gate{\mathcal{I}} & \qw &\\
& \vdots &  \vdots &  \vdots & \\
\qw & \gate{\mathcal{I}} & &\gate{\mathcal{I}}& \qw &
\end{quantikz}
    \caption{A depiction of $\mathcal{B}(\mathcal{H} , t, \epsilon, 0)$, that is equivalent to $e^{-i \mathcal{H} t}$.}
    \label{fig:basecircuitForwardEvolution}
\end{figure}

Fig.~\ref{fig:basecircuitForwardEvolution} shows that a key requirement of $\mathcal{B}\big( \mathcal{H}, t, \epsilon, 0 \big)$ is that it implements $e^{- i \mathcal{H} t}$, as is shown in Lemma~\ref{lem:EfisTimeEvolution}, which is stated without proof.
\begin{lemma}
    \label{lem:EfisTimeEvolution}
     For any spin Hamiltonian, $\mathcal{H}$, $\forall t \in \mathbb{R}$, $ \mathcal{B}\big(\mathcal{H}, t, \epsilon, 0 \big) = e^{- i \mathcal{H} t}$.
\end{lemma}
\subsection{Defining the Inverted Evolution, $\mathcal{B} (\mathcal{H}, t, \epsilon, 1)$}
\label{sec:AddingSingleToInvert}
We now turn to augmenting $\mathcal{B} (\mathcal{H}, t, \epsilon, 0)$ with single-qubit gates to produce $\mathcal{B} (\mathcal{H}, t, \epsilon, 1)$, which is equivalent to implementing $e^{i \mathcal{H} t / (L - 1)}$. Here $L \in \mathbb{N}$ is defined in terms of sets that need defining first. 

We start with the foundational one, $\mathbb{G}_{\mathcal{H}}$ in Def.~\ref{def:G}.
\begin{definition}
    \label{def:G}
    \underline{$\mathbb{G}_{\mathcal{H}}$} $\subseteq \big \{ q \sigma \text{ } \vert \text{ } \sigma \in \{ \mathcal{X}, \mathcal{Y}, \mathcal{Z}, \mathcal{I}\}^N, q \in \{ \pm 1, \pm i \} \big \}$ depends exclusively on $\mathcal{H} = \sum_{\mathbf{u} \in J} \bigg( c_{\mathbf{u}} \mathcal{P}_{\mathbf{u}} \bigg)$ and is defined by:
    \begin{align}
        \forall \mathbf{u} \in J \backslash \{ \Vec{0} \}, 
        \exists \mathcal{\sigma} \in \mathbb{G}_{\mathcal{H}} \text{ such that } \big \{ \mathcal{\sigma}, \mathcal{P}_{\mathbf{u}} \big \} = 0.
    \end{align}
\end{definition}

As shown in Lemma~\ref{lem:GExists} (in Appendix~\ref{sec:AppendixProofOfTheorem:inversionWorks}) for any Hamiltonian a corresponding $\mathbb{G}_{\mathcal{H}}$ always exists and can be found efficiently.
From $\mathbb{G}_{\mathcal{H}}$ we define a closely related set that we call $\mathbb{G}_{\mathcal{H}}'$ and define formally in Def.~\ref{def:G'}.
\begin{definition}
    \label{def:G'}
    \underline{$\mathbb{G}'_{\mathcal{H}}$} $\subseteq \{ \mathcal{X}, \mathcal{Y}, \mathcal{Z}, \mathcal{I}\}^N$ is defined as:
    \begin{align}
        \forall \mathcal{\sigma} \in \mathbb{G}_{\mathcal{H}}, \exists \mathcal{\sigma}' \in \mathbb{G}_{\mathcal{H}}'
        \text{ such that }
        \mathcal{\sigma}' \propto \mathcal{\sigma},
    \end{align}
    where $\propto$ denotes equality up to a factor in $\{\pm1,\pm i \}$. 
\end{definition}
Finally, we can define the quantity that began this chain of definitions, $L \in \mathbb{N}.$ 
\begin{definition}
    \label{def:L}
    \underline{$L$} $\in \mathbb{N}$ depends exclusively on the Hamiltonian, $\mathcal{H}$, being  simulated and is defined as $ L = \big \vert \mathbb{G}_{\mathcal{H}}' \big \vert$.
\end{definition}
    $L$ is independent of $M$, and $N$ only provides a bound on $L$ as $4^N$ bounds the number of terms in the Hamiltonian,  which in turn bounds the size of $\mathbb{G}'_{\mathcal{H}}$.

Using the sets and variables defined above, we present Algorithm~\ref{alg:addingSinglesForInvert} to generate $\mathcal{B}\big(\mathcal{H}, t, \epsilon, 1).$
Fig.~\ref{fig:basecircuitInversionSubcircuit} depicts $\mathcal{B}\big(\mathcal{H}, t, \epsilon, 1).$
\begin{figure}
    \centering
\begin{algorithm}[H]
$\mathbf{Input:}$ \\
$\bullet$ A Hamiltonian, $\mathcal{H}$, where $c_j$ is the coefficient of its $j$th term\\
$\bullet$ A evolution duration, $t \in \mathbb{R}$\\
$\bullet$ A maximum permissible additive error -- in terms of the diamond norm -- in the approximate time inversion, $\epsilon \in \mathbb{R}$
\vspace{0.3cm}
    \begin{enumerate}
        \item Calculate the $\mathbb{G}'_{\mathcal{H}}$ required to invert $\mathcal{H}$, as defined in Def.~\ref{def:G'}
        \item Calculate $L = \vert \mathbb{G}'_{\mathcal{H}} \vert$
        \item Calculate $M =  \dfrac{2 t^2  \big[ \sum_{j = 1} \big( \vert c_j \vert \big) \big]^2 L}{\epsilon (L - 1)} $
        \item Initialize a blank circuit, $\textit{circ}$
        \item Set $\text{timeStep} = 1$
        \item \While{\text{timeStep} $\leq M$}{
            \begin{enumerate}
                \item Choose $\mathcal{\sigma}$ uniformly at random from $\mathbb{G}'_{\mathcal{H}}$
                \item Add $\mathcal{\sigma} e^{-i t \mathcal{H} / M}  \mathcal{\sigma}$ to $\textit{circ}$
                \item $\text{timeStep } += 1$
            \end{enumerate}
        }
    \end{enumerate}
\vspace{0.1cm}
$\mathbf{Return}:$ $\textit{circ}$
\caption{$\mathcal{B}(\mathcal{H}, t, \epsilon, 1)$ construction algorithm
 \label{alg:addingSinglesForInvert}}
\end{algorithm}
\end{figure}

\begin{figure}[h]
    \centering
 \begin{quantikz}[row sep=0.3cm]
\qw & \gate{\mathcal{B}^{(k)}_1}  \gategroup[5,steps=3,style={inner
sep=2pt}]{Repeat $M$ times, indexing the loops by $k$} & \gate[wires=5, nwires=4]{e^{-i \mathcal{H} t / M}} & \gate{\mathcal{B}^{(k)}_1} & \qw &\\
\qw & \gate{\mathcal{B}^{(k)}_2} && \gate{\mathcal{B}^{(k)}_2} & \qw &\\
\qw & \gate{\mathcal{B}^{(k)}_3} && \gate{\mathcal{B}^{(k)}_3} & \qw &\\
 & \vdots &  \vdots &  \vdots &\\
\qw & \gate{\mathcal{B}^{(k)}_N} & &\gate{\mathcal{B}^{(k)}_N}& \qw &
\end{quantikz}
    \caption{A depiction of a subcircuit, $\mathcal{B}(\mathcal{H}, t, \epsilon, 1)$, that is approximately equal to $e^{i \mathcal{H} t/(L - 1)}$, where each $\mathcal{B}^{(k)}_j$ is an independent -- from all other $\mathcal{B}^{(k)}_j$ -- uniformly chosen element of the relevant $\mathbb{G}_{\mathcal{H}}'$.}
    \label{fig:basecircuitInversionSubcircuit}
\end{figure}
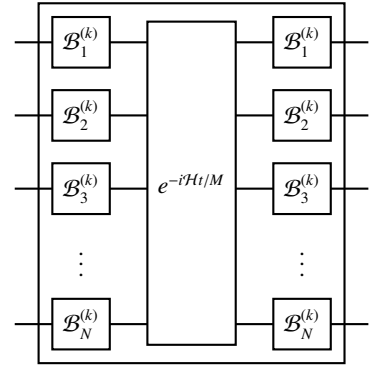

We now present the key result concerning $\mathcal{B}(\mathcal{H}, t, \epsilon, 1)$: Theorem~\ref{thm:inversioWorks},  proven in Appendix~\ref{sec:AppendixProofOfTheorem:inversionWorks}.

\begin{theorem}
    \label{thm:inversioWorks}
    For any Hamiltonian, $\mathcal{H}$, and $t \in \mathbb{R}$, there exists an efficiently computable $\mathbb{G}_{\mathcal{H}}'$ such that Algorithm~\ref{alg:addingSinglesForInvert} produces a $\mathcal{B}(\mathcal{H}, t, \epsilon, 1)$ that approximates $e^{i \mathcal{H} t/(L - 1)}$ , to within arbitrary additive error $\epsilon$ in the diamond norm.
\end{theorem}

\begin{proof}[Proof of Theorem~\ref{thm:alwaysInvertible}]
    $\mathcal{B}(\mathcal{H}, t, \epsilon, 0)$ (as in Fig.~\ref{fig:basecircuitForwardEvolution}) is the required subcircuit that, as shown in Lemma~\ref{lem:EfisTimeEvolution}, is equivalent to $e^{-i \mathcal{H}t}$.
    
Purely by adding single-qubit gates, $\mathcal{B}(\mathcal{H}, t, \epsilon, 0)$ can be transformed into $\mathcal{B}(\mathcal{H}, t, \epsilon, 1)$ (as in Fig.~\ref{fig:basecircuitInversionSubcircuit}), which is approximately equivalent to arbitrary additive error in the diamond norm to $e^{i \mathcal{H}t / (L-1)}$ (as shown in Theorem~\ref{thm:inversioWorks}).
    
The single-qubit gates added to $\mathcal{B}(\mathcal{H}, t, \epsilon, 0)$, to produce $\mathcal{B}(\mathcal{H}, t, \epsilon, 1)$, are uniformly random elements of $\mathbb{G}_{\mathcal{H}}'$, which is guaranteed to exist and be efficiently computable by Theorem~\ref{thm:inversioWorks}.
    
    By counting the resources in the proof of Theorem~\ref{thm:inversioWorks}, we can see that $\mathcal{B} (\mathcal{H}, t, \epsilon, 0)$ and $\mathcal{B} (\mathcal{H}, t, \epsilon, 1)$ require splitting the time evolution into $ O(t^2/\epsilon)$ parts, and $\mathcal{B} (\mathcal{H}, t, \epsilon, 1)$ uses $O(N t^2/\epsilon)$ single-qubit gates.
\end{proof}

\section{The Hybrid Digital-Analogue Model and Defining Accreditation}
\begin{definition}
\label{AQSDef}
    An \underline{analogue quantum simulation} takes as inputs:
    \begin{enumerate}
   \item The description of an initial product state $\ket{\psi_0}$,
     \item  A time-independent Hamiltonian, $\mathcal{H}$,
   \item A simulation duration, $t \in \mathbb{R}$,
\item  A set of single-qubit measurements, $\mathcal{M}$.

\end{enumerate}
    The simulation prepares $\ket{\psi_0}$ then applies the time evolution generated by $\mathcal{H}$, for the duration $t$, followed by  the measurements in $\mathcal{M}$. It returns, as the final output, the results of the measurements in $\mathcal{M}$.
\end{definition}

While an analogue quantum simulator has many uses, it does not suffice for accrediting analogue quantum simulations; which is the aim of this paper and defined formally in Def.~\ref{AAQSDef}. 
However, to perform, and define, accreditation protocols, we first require a model where analogue quantum simulation is augmented with the ability to apply single-qubit gates, defined in Def.~\ref{HQSDef}.
\begin{definition}
\label{HQSDef}
    A \underline{hybrid quantum simulation} (HQS) takes the four inputs of the analogue quantum simulation in Def.~\ref{AQSDef} and
    \begin{enumerate}
    \setcounter{enumi}{4}
    \item An ordered set, $\mathbb{S}$, of single-qubit\footnote{This definition is less demanding than the one in Ref.~\cite{jackson2023accreditation} by not allowing two-qubit gates and is a strict subset of the gate sequences allowed there. This moves HQSs closer to practicality as contemporary analogue simulators tend to either have, or intend to have, single-qubit gates but no many-qubit gates.} quantum gates with corresponding time-stamps $\{t_{\gamma} \text{ } \vert  \text{ }  \gamma \in \mathbb{S} \}$ denoting when each is applied. 
    \end{enumerate}
\end{definition}
As promised, we now define accredited analogue quantum simulations and accreditation protocols, in Def.~\ref{AAQSDef}.
\begin{definition}
\label{AAQSDef}
An \underline{accredited analogue quantum simulation} runs on a HQS. 
It takes all inputs of an analogue quantum simulation in Definition~\ref{AQSDef} and two parameters $\alpha, \theta \in [0,1)$.
It returns the same measurement outcomes as an analogue quantum simulation, and additionally an $\epsilon_{\text{VD}} \in [0,1)$ that upper bounds the ideal-actual variation distance (in Def. \ref{def:idealActualVarDist}) of the execution from which the measurement outcomes are obtained, with accuracy $\theta$ and confidence $\alpha$. 
\end{definition}

\section{Error model}
\label{ErrorModelSubsection}
\subsection{Redaction}
The error model used in this paper is expressed formally in terms of redaction, as defined in the following.
\begin{definition}
\label{def:redactedGates}
    A gate in a HQS is \underline{redacted} if the gate's position in the simulation (e.g. when it is applied and to which qubits) is specified but the operator it represents (e.g. if it is a Pauli $\mathcal{X}$ gate or Pauli $\mathcal{Z}$ gate) is not.
    
    Similarly, a HQS or circuit is redacted if all of its single-qubit gates (except the identity) are redacted. For an example of a redacted circuit, see Fig.~\ref{fig:redactionExample}.
    \begin{figure}
    \centering
    \includegraphics[width=0.29\textwidth]{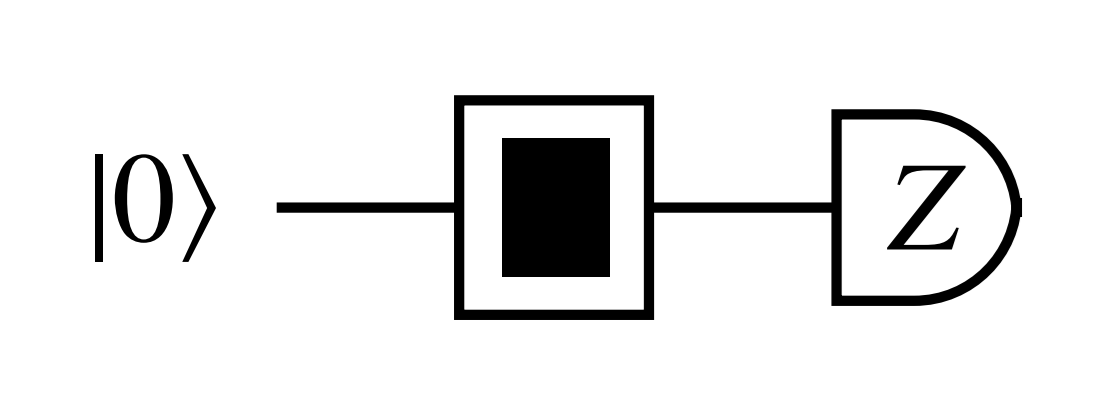}
    \caption{(Ref.\cite[Fig.1]{jackson2024accreditationlimitedadversarialnoise}) An example circuit with a redacted gate (the block in the middle of the circuit). Note how the location of the gate is depicted -- in terms of when it is applied and which qubits it acts on -- but what operation the gate represents is hidden.} \label{fig:redactionExample}
\end{figure}
\end{definition}
\begin{definition}
    If given a redacted circuit representing a HQS, the set of all HQSs that the given redacted HQS could possibly be, if the redactions were removed, is called its \underline{redaction class}.
\end{definition}

\begin{definition}
    Two given HQSs are said to be \underline{in the same redaction class} if there exists a redacted HQS such that both the given HQSs are in its redaction class.
    
    For example, all the circuits in Fig.~\ref{fig:redactionClassExample} are within the same redaction class as the redacted circuit in Fig.~\ref{fig:redactionExample} could be any one of the circuits in Fig.~\ref{fig:redactionClassExample} if its redaction were removed.
    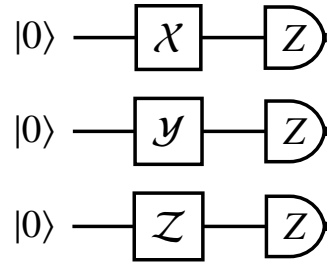
\begin{figure}
    \centering
 \begin{adjustbox}{width=0.25\textwidth}
     \begin{quantikz}
    \lstick{$\ket{0}$} & \gate{\mathcal{X}} & \meterD{Z}
    \end{quantikz}
    \end{adjustbox}
    \begin{adjustbox}{width=0.25\textwidth}
    \begin{quantikz}
    \lstick{$\ket{0}$} & \gate{\mathcal{Y}} & \meterD{Z}
    \end{quantikz}
    \end{adjustbox}
    \begin{adjustbox}{width=0.25\textwidth}
    \begin{quantikz}
    \lstick{$\ket{0}$} & \gate{\mathcal{Z}} & \meterD{Z}
    \end{quantikz}
    \end{adjustbox}
    \caption{(Ref.\cite[Fig.2]{jackson2024accreditationlimitedadversarialnoise}) Example circuits within the same redaction class, where the redaction class corresponds to the redacted circuit in Fig.~\ref{fig:redactionExample}. Note that the circuit in Fig.~\ref{fig:redactionExample} could be any of these if the redaction were removed.}
     \label{fig:redactionClassExample}
\end{figure}
\end{definition}
\begin{definition}
    Let $\mathbb{D}$ be a redaction class and $\mathcal{C} \in \mathbb{D}$. If a subset of the non-identity single-qubit gates in $\mathcal{C}$ are replaced with identity gates and then the remaining non-identity single-qubit gates are redacted, the redaction class of the resulting HQS is a \underline{reduced redaction class of $\mathbb{D}$}. 
\end{definition}
 \subsection{Statement of Error Model}
\label{sec:ErrorModelStatement}
We model \emph{any} erroneous implementation of \emph{any part} of a HQS as its error-free implementation followed, or preceded, by an error operator such that:
\begin{itemize}
    \item[E1.] The error operator is a completely positive trace preserving (CPTP) map on the simulator and its environment. It is independent of the error in any other simulation. 	
     \item[E2:] Let $\mathbb{D}$ be a redaction class and $\mathbb{D}'$ be one of its reduced redaction classes. Then, $\forall \mathcal{C} \in \mathbb{D}$, $\mathcal{C}' \in \mathbb{D}'$, for any executions of $\mathcal{C}$ and $\mathcal{C}'$; $\Tilde{C}$ and $\Tilde{C'} $, respectively:
    \begin{align}
        \nu \big( \Tilde{C} \big) \geq \nu \big( \Tilde{C'} \big).
    \end{align}
    \item[E3:] The error in an execution of a simulation may depend on which redaction class a simulation is in, but is independent of which HQS within that class is being executed.
\end{itemize}

\subsection{Justification of the Error Model}

This error model is more general than most employed in, for instance, randomised benchmarking and Hamiltonian learning.
In particular, errors included in our model need not be time-independent or Markovian, and SPAM errors are not treated any differently~\cite{2020AnalogueBenchmarking}.
Our model includes physical error processes such as spontaneous emission, crosstalk, and particle loss.
It also captures both fast and slow noise processes such as laser fluctuations and temperature variations or degradation of device performance over implementations respectively~\cite{Shaffer2021}.
This is enabled by the forgoing the assumption of identically distributed errors, representing an advance over Ref.~\cite{jackson2023accreditation}.

We now justify each assumption listed in Sec.~\ref{sec:ErrorModelStatement}. 

\begin{enumerate}
    \item[E1:] Any map from and to density matrices is a CPTP map. Hence, any erroneous operation must be a CPTP map. 
    Let $\mathcal{U}$ denote a unitary operator and $\mathcal{A}$ denote the CPTP map representing an erroneous implementation of $\mathcal{U}$. Letting $\mathcal{A}' =  \mathcal{A} \mathcal{U}^{\dagger}$ and $\mathcal{A}'' = \mathcal{U}^{\dagger} \mathcal{A}$:
\begin{align}
    \label{eqn:justCPTPEquation}
    \mathcal{A}
    &=
    \mathcal{A} \mathcal{U}^{\dagger} \mathcal{U}
    =
    \mathcal{A}' \mathcal{U}\\
    &\hspace{1cm} \text{ and } \nonumber\\
    \label{eqn:justCPTPEquation2}
    \mathcal{A} 
    &=
    \mathcal{U} \mathcal{U}^{\dagger} \mathcal{A}
    =
    \mathcal{U} \mathcal{A}''.
\end{align}
Eqn.~\ref{eqn:justCPTPEquation} and Eqn.~\ref{eqn:justCPTPEquation2} each represent the correct implementation of $\mathcal{U}$ followed or preceded, respectively, by a CPTP map.

\item[E2:] The essence of E2 is that implementing a gate results in noise having a greater effect than if we had not implemented the gate. This follows from the empirical observation that a greater number of gates in a circuit typically results in a greater error occurring (as can be seen in Ref.~\cite[Fig.~1(b) and Fig.~7(a)]{2021Sams}).

\item[E3:] E3 is a well established assumption~\cite{Erhard2019, PhysRevA.92.060302, PhysRevLett.114.140505, PhysRevA.87.062119, PhysRevLett.106.230501, Dahlhauser_2021, PhysRevLett.109.070504, Harper2020, ferracin2022efficiently}. It is more often stated as: single-qubit gates experience gate-independent error. This
follows from single-qubit gates typically being the least error-prone components in most quantum hardware platforms~\cite{doi:10.1126/science.1145699, PhysRevLett.113.220501, Wright2019, Arute2019}. It is a standard assumption in the prior accreditation protocols~\cite{Ferracin_2019, Ferracin_2021, jackson2023accreditation}. E3 may be relaxed to allow for error that depends weakly on single-qubit gates~~\cite[Appendix 2]{2021Sams}.
\end{enumerate}

\subsection{A Simplification of Error Considerations}

Throughout this paper we make use of Lemma~\ref{simplerErrorLem}, akin to one proved in Ref.~\cite[Lemma 1]{jackson2023accreditation}.
Lemma~\ref{simplerErrorLem} is not proven here due to this similarity.

\begin{lemma}
\label{simplerErrorLem}    
Any HQS affected by errors obeying E1-E3 can be considered as if all single-qubit gates are error-free and all remaining error is independent of the single-qubit gates.
\end{lemma}

\section{Improved Analogue Accreditation Protocol}
\label{sec:protocolPresent}
This section presents our improved protocol for the accreditation of analogue quantum simulations.

\subsection{Defining Trap and Target Simulations}
\label{sec:defTrapAndTarg}
In presenting our trap and targets, we wish to inherit much of the analogue accreditation protocol in Ref.~\cite{jackson2023accreditation}. We do this via Lemma~\ref{lem:trapInherited}.
\begin{lemma}
    \label{lem:trapInherited}
    Assuming E1-3, given a subcircuit equivalent to the identity -- referred to hereafter as $\mathcal{J}(\mathcal{H}, t, \epsilon, 1)$ -- and another  -- referred to hereafter as $\mathcal{J}(\mathcal{H}, t, \epsilon, 0)$ -- that differs from the first only in its single-qubit gates (with $\mathcal{J}(\mathcal{H}, t, \epsilon, 1)$ having more non-identity single-qubit gates), but is equivalent to a time evolution, $e^{- i\mathcal{H}t}$, two simulations -- known as the trap\footnote{which includes the use of randomly chosen gates.} and the target simulations, respectively -- can be constructed such that:\\
    \underline{\emph{The trap simulation:}}
    \begin{enumerate}
        \item reduces all error occurring to stochastic Pauli error.
        \item detects Pauli error with probability at least $1/2$.
    \end{enumerate}
     \underline{\emph{The target simulations:}}
    \begin{enumerate}
        \item are equivalent to applying $e^{- i\mathcal{H}t}$ to a specified single-qubit product state before performing specified single-qubit measurements, $\mathcal{M}$.
        \item experiences error such that the ideal-actual (as defined in Def.~\ref{def:idealActualVarDist}) variation distance in executing it is always less than that in executing a trap. 
    \end{enumerate}
\end{lemma}
\begin{proof}
   Instead of a full proof we sketch the overall approach:
\begin{enumerate}
    \item As in Fig.~\ref{fig:TrapAndTestDepiction}, in the traps, the protocol in Ref.~\cite{jackson2023accreditation} applies single-qubit gates: a uniformly random Pauli and a Hadamard gate, with probability $0.5$, on each qubit before and after $\mathcal{J}(\mathcal{H}, t, \epsilon, 1)$.
    \item $\mathcal{J}(\mathcal{H}, t, \epsilon, 1)$ is equivalent to the identity, so any operator applied when $\mathcal{J}(\mathcal{H}, t, \epsilon, 1)$ is implemented \emph{is} the error operator.
    \item The uniformly random Pauli gates reduce\footnote{as in Ref.~\cite[Appendix A]{Ferracin_2019}.} the error occurring in the implementation of $\mathcal{J}(\mathcal{H}, t, \epsilon, 1)$, and the single-qubit gates themselves, to stochastic Pauli error.
    \item The stochastic Hadamard gates mean any Pauli error is detectible with probability at least $0.5$~\cite[Lemma 3]{jackson2023accreditation}, as the traps give a single definite outcome if there is no error~\cite[Lemma 6]{jackson2023accreditation}.
    \item Multiple runs of the trap simulation estimate the probability of error occurring and Ref.~\cite[Lemma 8]{jackson2023accreditation} provide a bound on the ideal-actual variation distance in traps.
    \item By Ref.~\cite[Lemma 2]{jackson2023accreditation}, the ideal-actual variation distance is greater in traps than in the target simulation.
\end{enumerate}
\end{proof}
\begin{corollary}
    \label{cor:redactionClassesTrapsAndTargets}
    In any accreditation protocol using Lemma~\ref{lem:trapInherited} to generate trap and target simulations\footnote{along with specified $\mathcal{J}(\mathcal{H}, t, \epsilon, 0)$ and $\mathcal{J}(\mathcal{H}, t, \epsilon, 1)$.}, the trap simulations are all in the same redaction class and the target simulation is in a reduced redaction class of that redaction class.
\end{corollary}
In light of Lemma~\ref{lem:trapInherited}, the rest of Sec.~\ref{sec:defTrapAndTarg} is devoted to constructing $\mathcal{J}(\mathcal{H}, t, \epsilon, 0)$ and $\mathcal{J}(\mathcal{H}, t, \epsilon, 1)$.

\begin{definition}
\label{def:vanishBlockDef}
A modified $j$-vanishing block is a subcircuit of the form in Fig. \ref{fig:vanishingBlockDef}, where $t_1'$ and $t_2'$ are related as in Def.~\ref{def:t1Primet2Prime}. 
Due to Theorem~\ref{thm:alwaysInvertible}, a modified $j$-vanishing block exists for any Hamiltonian and $j \in \{0, 1\}$.
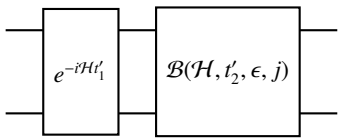
\begin{figure}[h]
    \centering
\begin{center}
\begin{quantikz}
 & \gate[wires=2]{e^{-i \mathcal{H}t_1'}}  & \gate[wires=2]{\mathcal{B}(\mathcal{H}, t'_2, \epsilon, j)} & \qw\\
 &&& \qw
\end{quantikz}
\end{center}
    \caption{Depiction of a modified $j$-vanishing block, $\mathcal{J}(\mathcal{H}, t_1' + t_2', \epsilon, j)$, where $\mathcal{H}$ is an arbitrary spin Hamiltonian, $ t_1', t_2' \in \mathbb{R}$, $\mathcal{B}(\mathcal{H}, t'_2, \epsilon, j)$ is as in Fig.~\ref{fig:basecircuitInversionSubcircuit}, and $j \in \{ 0, 1\}$.}
    \label{fig:vanishingBlockDef}
\end{figure}
\end{definition}
Def.~\ref{def:vanishBlockDef} has made use of Def.~\ref{def:t1Primet2Prime} in splitting a
time evolution (of duration $t \in \mathbb{R}$) into two time evolutions: the first of duration $t_1' \in \mathbb{R}$ and the second of duration $t_2' \in \mathbb{R}$.
\begin{definition}
\label{def:t1Primet2Prime}
    For any given $t \in \mathbb{R}$ and $L \in \mathbb{N}$, let
    \begin{align}
        t'_1  = \dfrac{t}{L}
        \text{ and }
        t'_2 = \dfrac{(L - 1)t}{L}.
    \end{align}
\end{definition}
The most important property of modified vanishing blocks -- presented in Lemma~\ref{lem:IdealJLemma} -- follows from two properties of $t'_1$ and $t_2'$, proven in Lemma~\ref{lem:t1t2Sum}.
\begin{lemma}
    \label{lem:t1t2Sum}
    If $t_1', t_2' \in \mathbb{R}$ are constructed as prescribed in Def.~\ref{def:t1Primet2Prime} using any given $t \in \mathbb{R}$ and  $ L \in \mathbb{N}$, then:
        \begin{enumerate}
            \item $t'_1 + t'_2 = t$,
            \item $t'_1 - \dfrac{t'_2}{L - 1} = 0$.
        \end{enumerate}
\end{lemma}
\begin{proof}Lemma~\ref{lem:t1t2Sum} follows from substituting in the definitions of $t'_1$ and $t'_2$ from Def.~\ref{def:t1Primet2Prime}:
\begin{enumerate}
        \item $t'_1 + t'_2
        =
        \dfrac{t}{L} + \dfrac{(L - 1)t}{L}
        =
        \dfrac{(L - 1 + 1)t}{L}
        =
        t$,
        \item $t'_1 - \dfrac{t'_2}{L - 1} 
        =
        \dfrac{t}{L} - \dfrac{1}{L - 1} \dfrac{(L - 1)t}{L}
        =
        \dfrac{t}{L} -  \dfrac{t}{L}
        =
        0$.
        \end{enumerate}
\end{proof}
\begin{lemma}
\label{lem:IdealJLemma}
    The errorless implementation of a modified $j$-vanishing block, $\mathcal{J}(\mathcal{H}, t, \epsilon, j)$, is equivalent (but only up to additive error $\epsilon$ in the diamond norm when $j = 1$) to $e^{-i (1 - j) \mathcal{H}t}$, where $t \in \mathbb{R}$ is as can be inferred from $t_1'$ and $t_2'$ in Def.~\ref{def:t1Primet2Prime}.
\end{lemma}
\begin{proof}
    By Lemma~\ref{lem:EfisTimeEvolution}, when there is no error, $\mathcal{B}(\mathcal{H}, t_2', \epsilon, 0) = e^{-i \mathcal{H}t_2'}$, hence a $0$-vanishing block, $\mathcal{J}(\mathcal{H}, t, \epsilon, 0)$, may be expressed as:
    \begin{align}
         e^{-i \mathcal{H}t_1'}  e^{-i \mathcal{H}t_2'}
         =
         e^{-i \mathcal{H} (t_1' + t_2')}
         =
         e^{-i \mathcal{H} t}.
    \end{align}
    Likewise, by Theorem~\ref{thm:inversioWorks}, $\mathcal{B}(\mathcal{H}, t_2', \epsilon, 1) \approx e^{i \mathcal{H}t_2'/ [L+1]}$, hence a $1$-vanishing block, $\mathcal{J}(\mathcal{H}, t, \epsilon, 1)$, may be approximated as:
    \begin{align}
        e^{-i \mathcal{H}t_1'} e^{i \mathcal{H}t_2'/ [L+1]}
        =
        e^{-i \mathcal{H}(t_1' - t_2'/ [L+1])}
        =
        e^{0}
        =
        \mathcal{I}.
    \end{align}
\end{proof}
Due to Lemma~\ref{lem:IdealJLemma}, all the requirements of Lemma~\ref{lem:trapInherited} are met. This provides the required trap and target simulations using the results of Ref.~\cite{jackson2023accreditation}. 
These are depicted in Fig.~\ref{fig:TrapAndTestDepiction} [top] -- for the target simulation -- and Fig.~\ref{fig:TrapAndTestDepiction} [bottom] -- for the trap simulations, and formal algorithms for their construction are given in Algorithm~\ref{alg:TargSimGen} and Algorithm~\ref{alg:TrapSimGen}.
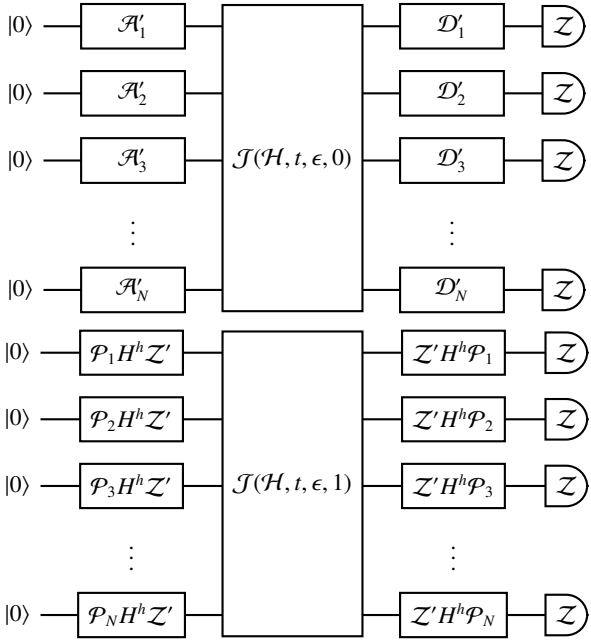
\begin{figure}[t]
    \centering
{\begin{quantikz}[row sep=0.3cm]
\lstick{$\ket{0}$} & \gate[][1.38cm]{\mathcal{A}^{\prime}_1} & \gate[wires=5, nwires=4]{\mathcal{J}(\mathcal{H}, t, \epsilon, 0)} & \gate[][1.38cm]{\mathcal{D}^{\prime}_1} & \meterD{\mathcal{Z}}\\
\lstick{$\ket{0}$} & \gate[][1.38cm]{\mathcal{A}^{\prime}_2}  &  & \gate[][1.38cm]{\mathcal{D}^{\prime}_2}&  \meterD{\mathcal{Z}}\\
\lstick{$\ket{0}$} & \gate[][1.38cm]{\mathcal{A}^{\prime}_3} &   &\gate[][1.38cm]{\mathcal{D}^{\prime}_3} & \meterD{\mathcal{Z}}\\
 & \vdots &  \vdots & \vdots & & \\
\lstick{$\ket{0}$} & \gate[][1.38cm]{\mathcal{A}^{\prime}_N} &   & \gate[][1.38cm]{\mathcal{D}^{\prime}_N} & \meterD{\mathcal{Z}}
\end{quantikz}}
\qquad
{ \begin{quantikz}[row sep=0.3cm]
\lstick{$\ket{0}$} & \gate{\mathcal{P}_1 H^h \mathcal{Z}^{\prime}} & \gate[wires=5, nwires=4]{\mathcal{J}(\mathcal{H}, t, \epsilon, 1)} & \gate{\mathcal{Z}^{\prime} H^h \mathcal{P}_1} & \meterD{\mathcal{Z}}\\
\lstick{$\ket{0}$} & \gate{\mathcal{P}_2 H^h \mathcal{Z}^{\prime}}  &  & \gate{\mathcal{Z}^{\prime} H^h \mathcal{P}_2}&  \meterD{\mathcal{Z}}\\
\lstick{$\ket{0}$} & \gate{\mathcal{P}_3 H^h \mathcal{Z}^{\prime}} &   &\gate{\mathcal{Z}^{\prime} H^h \mathcal{P}_3} & \meterD{\mathcal{Z}}\\
 & \vdots & & \vdots &  & \\
\lstick{$\ket{0}$} & \gate{\mathcal{P}_N H^h \mathcal{Z}^{\prime}} &  & \gate{\mathcal{Z}^{\prime} H^h \mathcal{P}_N} & \meterD{\mathcal{Z}}
\end{quantikz}}
    \caption{A target (top) and trap (bottom) simulation corresponding to a simulation according to Hamiltonian, $\mathcal{H}$, for a time duration $t$ applied to an initial state, $\ket{\psi_0}= \otimes_{j = 1}^N \big( \mathcal{A}^{\prime}_j \big) \ket{0}^{\otimes N}$, before the measurements $\mathcal{M}= \otimes_{j = 1}^N \big( \mathcal{D}^{\prime}_j \big) \mathcal{Z}^{\otimes N}$ are applied. $H$ denotes a Hadamard gate, $h \in \{0,1\}$ is a uniformly random bit, $\mathcal{Z}^{\prime}$ denotes applying a Pauli Z gate with probability $0.5$ (each instance in the circuit is independent), and $\mathcal{P}_j$ is a single-qubit Pauli gate, chosen uniformly at random independently for each $j$.}
    \label{fig:TrapAndTestDepiction}
\end{figure}
\begin{figure}
    \centering
\begin{algorithm}[H]
$\mathbf{Input:}$ \\
$\bullet$ A Hamiltonian, $\mathcal{H}$\\
$\bullet$ A evolution duration, $t \in \mathbb{R}$\\
$\bullet$ An initial state, $\vert \psi_0 \rangle =  \otimes_{j = 1}^N \big( \mathcal{A}^{\prime}_j \big) \ket{0}^{\otimes N}$, where each $\mathcal{A}^{\prime}_j$ is a single-qubit gate\\ 
$\bullet$ A vector of measurements, $\mathcal{M}$, equivalent to applying the single-qubit gates, $\otimes_{j = 1}^N \big( \mathcal{D}^{\prime}_j \big)$, then measuring in the computational basis
\vspace{0.3cm}
    \begin{enumerate}
        \item Initialize a blank circuit, $\textit{circ} $
        \item Prepare the state $\vert 0 \rangle^{\otimes N}$ as the input state to $\textit{circ}$
        \item Apply $\otimes_{j = 1}^N \big( \mathcal{A}^{\prime}_j \big)$ to $\textit{circ}$
        \item Calculate $t_1'$ and $t_2'$, as defined in Def.~\ref{def:t1Primet2Prime}
        \item Apply $e^{-i\mathcal{H}t_1'}$ to $\textit{circ}$
        \item Apply $\mathcal{B}(\mathcal{H}, t_2', \epsilon, 0)$ to $\textit{circ}$, using Algorithm~\ref{alg:forwardTime}
        \item Apply $\otimes_{j = 1}^N \big( \mathcal{D}^{\prime}_j \big)$ to $\textit{circ}$
        \item Measure each qubit in the $Z$-basis
    \end{enumerate}
\vspace{0.1cm}
$\mathbf{Return}:$ $\textit{circ}$
\caption{Target simulation construction algorithm
 \label{alg:TargSimGen}}
\end{algorithm}
\end{figure}

\begin{Note}
    In Algorithm~\ref{alg:TrapSimGen}, $H$ is used to denote a Hadamard gate. It is not written in calligraphic letters (i.e. as $\mathcal{H}$) to avoid confusion with a Hamiltonian.\\
    Additionally, $H^h$ denotes $H$ if $h=1$ and $\mathcal{I}$ if $h=0$. 
\end{Note}

\begin{figure}
    \centering
\begin{algorithm}[H]
$\mathbf{Input:}$ \\
$\bullet$ A Hamiltonian, $\mathcal{H}$\\
$\bullet$ A evolution duration, $t \in \mathbb{R}$\\
$\bullet$ A number of qubits $\mathcal{H}$ acts on , $N \in \mathbb{N}$\\ 
$\bullet$ A maximum permissible additive error -- in terms of the diamond norm -- in the approximate time inversion, $\epsilon \in \mathbb{R}$
\vspace{0.3cm}
    \begin{enumerate}
        \item Initialize a blank circuit, $\textit{circ} $
        \item Prepare the state $\vert 0 \rangle^{\otimes N}$ as the input state to $\textit{circ}$
        \item Choose $h \in \{ 0,1 \}$ uniformly at random
        \item For qubit $j$:
        \begin{enumerate}
            \item With probability $0.5$, apply a Pauli $\mathcal{Z}$ gate to qubit $j$
            \item Choose Pauli gate $\mathcal{P}_j$ uniformly at random
            \item Apply $\mathcal{P}_jH^h$ to qubit $j$
        \end{enumerate}
        \item Calculate $t_1'$ and $t_2'$, as defined in Def.~\ref{def:t1Primet2Prime}
        \item Apply $e^{-i\mathcal{H}t_1'}$ to $\textit{circ}$
        \item Apply $\mathcal{B}(\mathcal{H}, t_2', \epsilon, 1)$ to $\textit{circ}$, with additive error in the inverted time evolution bounded by $\epsilon$, using Algorithm~\ref{alg:addingSinglesForInvert}
        \item For qubit $j$:\\
        \begin{enumerate}
            \item Apply $H^h\mathcal{P}_j$ to qubit $j$
            \item With probability $0.5$ apply a Pauli $\mathcal{Z}$ gate to qubit $j$
        \end{enumerate}
        \item Measure each qubit in the $\mathcal{Z}$-basis
    \end{enumerate}
\vspace{0.1cm}
$\mathbf{Return}:$ $\textit{circ}$
\caption{Trap simulation construction algorithm
 \label{alg:TrapSimGen}}
\end{algorithm}
\end{figure}

\subsection{From Trap and Target Simulations to an Accreditation Protocol}
\label{sec:Trap2Targ}
\begin{lemma}
    \label{lem:twiceProbBound}
    Assuming E1-3, doubling the average probability of error over a set of trap simulations -- that are generated as in Lemma~\ref{lem:trapInherited} -- provides an upper bound on the ideal-actual variation distance of the corresponding target simulation.
\end{lemma}
\begin{proof}
    Let $\mathbb{K} = \big \{ \mathcal{C}_j \big \}$ denote a set of trap simulations, indexed by $j$. As noted in Lemma~\ref{lem:trapInherited}, every individual trap detects any stochastic error occurring in it with probability at least $0.5$. Therefore, the probability of error occurring is at most twice the probability that an individual trap detects error. The error assumptions E1-3 allow different traps, in $\mathbb{K}$, to have different probabilities of error but define $\langle P \rangle_{\mathbb{K}}$ as the average probability of error over $\mathbb{K}$.
    Due to Lemma~\ref{lem:trapInherited}, the traps we use in this protocol inherit from Ref.~\cite[Lemma 8]{jackson2023accreditation} that: 
    $\forall \mathcal{C}_j \in \mathbb{K}$,
    \begin{align}
        \nu \big( \Tilde{\mathcal{C}}_j \big)
        \leq
        2 \textit{Prob}_e \big( \Tilde{\mathcal{C}}_j  \big),
    \end{align}
    where $\textit{Prob}_e$ is the probability error occurs in the argument execution.
    Therefore, using Corollary~\ref{cor:redactionClassesTrapsAndTargets} and E2:
    \begin{align}
         \nu \big( \Tilde{\mathcal{C}}_{\mathrm{targ}} \big)
         \leq
        \dfrac{1}{\vert \mathbb{K} \vert} \sum_{j = 1}^{\vert \mathbb{K} \vert} \bigg( \nu \big( \Tilde{\mathcal{C}}_j \big) \bigg)
        \leq
        \dfrac{1}{\vert \mathbb{K} \vert} \sum_{j = 1}^{\vert \mathbb{K} \vert} \bigg(
        2 \textit{Prob}_e \big( \Tilde{\mathcal{C}}_j  \big) \bigg)
        =
        2 \langle P \rangle_{\mathbb{K}},
    \end{align}
    where $\nu \big( \Tilde{\mathcal{C}}_{\mathrm{targ}} \big)$ is the ideal-actual variation distance (as in Eqn.~\ref{def:idealActualVarDist}) of the execution of the corresponding (to the traps) target simulation, and $\textit{Prob}_e$ is the probability of error occurring in its argument simulation execution.
\end{proof}
Through the use of $\bigg \lceil \dfrac{2}{\theta^2} \ln{\bigg( \dfrac{2}{1 - \alpha} \bigg)} \bigg \rceil + 1$ traps -- such as the set $\mathbb{K}$ in Lemma~\ref{lem:twiceProbBound} -- the average probability that a trap in $\mathbb{K}$ detects error, $\langle P \rangle_{\mathbb{K}}$, can be estimated to within additive error $\theta$, with confidence $\alpha$, as shown in Ref.~\cite[The Proof of Theorem~1]{jackson2023accreditation}.
This final technique completes the accreditation protocol, which is presented -- formally --  in Protocol~\ref{fullProtocolSketch}.

The main result of this paper, that the presented Protocol~\ref{fullProtocolSketch} functions as an accreditation protocol, as defined in Def.~\ref{AAQSDef}, is Theorem~\ref{finalTheorem}.

\begin{theorem}
\label{finalTheorem}
    Assuming E1-3, Protocol~\ref{fullProtocolSketch} performs accredited analogue simulation as per Def.~\ref{AAQSDef}, using $N_{\textrm{tr}}$ trap simulations, where
    \be
        N_{\mathrm{tr}} = \bigg \lceil \dfrac{2}{\theta^2} \ln{\bigg( \dfrac{2}{1 - \alpha} \bigg)} \bigg \rceil + 1.
    \ee
\end{theorem} 
\begin{proof}
    Protocol~\ref{fullProtocolSketch} uses $\bigg \lceil \dfrac{2}{\theta^2} \ln{\bigg( \dfrac{2}{1 - \alpha} \bigg)} \bigg \rceil + 1$ trap simulations (which follows from Hoeffding's inequality~\cite{doi:10.1080/01621459.1963.10500830}) -- as per Ref.~\cite[The Proof of Theorem~1]{jackson2023accreditation} -- to estimate the average probability of error of those trap simulation executions to arbitrary additive error, $\theta$, with arbitrary confidence, $\alpha$. Due to Lemma~\ref{lem:twiceProbBound} (and assuming E1-3), doubling this estimate upper bounds the ideal-actual variation difference in the execution of the target simulation (as defined in Lemma~\ref{lem:trapInherited}) executed amongst the trap simulation (as defined in Lemma~\ref{lem:trapInherited}) executions in Protocol~\ref{fullProtocolSketch}. 
 Due to Theorem~\ref{thm:alwaysInvertible}, the required simulations can always be constructed and executed.
\end{proof}
\begin{figure}
    \centering
\begin{algorithm}[H]

\SetAlgorithmName{Protocol}{protocol}{List of Protocols}

$\mathbf{Input:}$ \\
$\bullet$ A Hamiltonian, $\mathcal{H}$\\
$\bullet$ A real evolution time, $t \in \mathbb{R}$\\
$\bullet$ An initial state, $\vert \psi_0 \rangle$\\ 
$\bullet$ A vector of measurements for after the time evolution, $\mathcal{M}$\\
$\bullet$ A maximum permissible additive error -- in terms of the diamond norm -- in the approximate time inversion, $\epsilon \in \mathbb{R}$\\
$\bullet$ A desired confidence, $\alpha$\\
$\bullet$ A desired accuracy of the output bound, $\theta$\\
\vspace{0.3cm}
 \begin{enumerate}
     \item Calculate the required number of traps $N_{\textrm{tr}} = \bigg \lceil \dfrac{2}{\theta^2} \ln{\bigg( \dfrac{2}{1 - \alpha} \bigg)} \bigg \rceil + 1$
     \item Pick uniformly at random an integer between $1$ and $N+1$ to be the index of the target simulation 
     \item For $i=1$ to $N_{\textrm{tr}}+1$
     \begin{enumerate}
         \item If simulation $i$ is the target simulation:
            \begin{enumerate}
                \item Generate the target simulation using Algorithm~\ref{alg:TargSimGen}
                \item Execute the target simulation just generated
                \item Record the measurement outcomes
            \end{enumerate}
        \item If simulation $i$ is a trap simulation:
            \begin{enumerate}
                \item Generate a trap simulation using Algorithm~\ref{alg:TrapSimGen}
                \item Execute the trap simulation just generated
                \item Record if the trap gave the correct output
            \end{enumerate}
     \end{enumerate}
     \item Calculate   $\epsilon_{\text{VD}}
    = 
    2\dfrac{\text{Number of correct traps}}{N_{\textrm{tr}}+1}$
\end{enumerate}

$\mathbf{Return}:$ Target simulation measurement outcomes and $\epsilon_{\text{VD}}$.
\caption{Analogue Accreditation Protocol
 \label{fullProtocolSketch}}
\end{algorithm}
\end{figure}

\section{Establishing Quantum Advantage in Experiments}
\label{sec:quantumAdvantage}

An objective driving the development of verification protocols is the ambition to demonstrate quantum advantage~\cite{Preskill:2012tg, Harrow_2017}, as any experiment to demonstrate quantum advantage will require verification. This is because without verifying that the quantum device has correctly performed the task considered classically intractable, no advantage has been demonstrated. Therefore, we now discuss both the prospects for quantum advantage based on simulation and the potential utility of our analogue accreditation protocol for verifying that advantage.

The starting point in the search for quantum advantage must be complexity theory. 
Refs.~\cite{PhysRevX.8.021010,ringbauer2024verifiablemeasurementbasedquantumrandom} show that simulating the time evolution of the Hamiltonian $\mathcal{H}_A = \sum_{(i,j) \in \mathbf{E}} \big( J_{i,j} \mathcal{Z}_i \mathcal{Z}_j\big) - \sum_{i \in \mathbf{V}} \big( h_{i} \mathcal{Z}_i\big)$ -- where the graph $(\mathbf{E}, \mathbf{V})$ specifies a square lattice and $J_{i,j}, h_i \in \mathbb{R}$ (for all indices) -- is classically intractable if the ideal-actual variation distance of the simulation's execution, under certain complexity theoretic and other assumptions such as the average-case hard fraction of 1 in 1000, is less than 0.292~\cite{Fujii_2017}.
Ref.~\cite{liu2024efficientlyverifiablequantumadvantage} takes a similar approach to demonstrating advantage and provides further simplifications, including reducing the Hamiltonian to $\mathcal{H}_A = \sum_{(i,j) \in \mathbf{E}} \big( J_{i,j} \mathcal{Z}_i \mathcal{Z}_j\big)$. 
The latter is particularly amenable to the original accreditation protocol~\cite{jackson2023accreditation} is particularly well suited -- in some ways. The fact that this simulation maintains its classical intractability when experiencing error, provided the ideal-actual variation distance  (i.e. $\nu(\Tilde{\mathcal{C}})$, where $\Tilde{\mathcal{C}}$ is the execution of the simulation) of its execution remains below a threshold, makes it well suited to practical quantum advantage demonstrations as quantum simulators will inevitably experience noise in the near-term.

That the requirement for quantum advantage in the experiment proposed in Ref.~\cite{PhysRevX.8.021010} is expressed in terms of the ideal-actual variation distance additionally makes the analogue accreditation protocol presented herein well suited for use in verifying it, as our protocol returns exactly this metric about the quality of a simulation's execution. This experiment also requires the ability to prepare many different product states and measure in a variety of single-qubit bases, which may alternatively be expressed as requiring single-qubit gates.

Therefore, the only additional requirement for using our accreditation protocol in a quantum advantage demonstration based on the complexity results of Ref.~\cite{PhysRevX.8.021010}, beyond those required for the demonstration without verification, is the ability to pause the implementation of time evolutions in order to apply single-qubit gates.

\section{Discussion}
\label{discussionSec}
We have presented a protocol for accrediting the outputs of analogue simulations. It is an updated and improved version of the protocol presented in Ref.~\cite{jackson2023accreditation}: the key improvements being the ability to invert arbitrary spin Hamiltonians -- without using universal Hamiltonians~\cite{zhou2021strongly} -- and relaxing the assumption that simulations all experience identical error. Additionally, the resource requirements of our protocol remain reasonable: the total duration of time evolutions in trap and target simulations are exactly as if the evolution were performed without accreditation, and the number of additional single-qubit gates scale linearly in all input parameters; no extra qubits are required; the number of trap simulations required is quadratic in the reciprocal of the required accuracy of the bound on the variation distance the protocol outputs. Consequently, our protocol can be implemented on extant programmable hybrid analogue-digital quantum simulators~\cite{Browaeys2020,Bluvstein2022}.

Our work leaves several potential avenues for further improvement. While we believe the assumption, in E1, that error is represented by CPTP maps to be well-justified, the independence of the CPTP maps in different simulations is suspect and hence it may be beneficial to either develop methods to remove this assumption, or confirm that this independence assumption holds.
An easy improvement -- that has already been addressed~\cite[Appendix 2]{2021Sams} -- is that E3 may be relaxed to allow for error that depends weakly on single-qubit gates. The results shown in Ref.~\cite{2021Sams} regarding this directly transfer to apply to this protocol.

We do not believe that E2 may be substantially relaxed without significant changes to the rest of the protocol. 

\section{Acknowledgements}
We thank Peter Zoller, Andrew Daley, Jonathan Pritchard, and Theodoros Kapourniotis for useful conversations.
This work was supported, in part, by a EPSRC IAA grant (G.PXAD.0702.EXP), the UKRI ExCALIBUR project QEVEC (EP/W00772X/2),
and the Hub for Quantum Computing via Integrated and Interconnected Implementations (QCI3) (EP/Z53318X/1).

\bibliography{References}

\begin{thebibliography}{49}%
\makeatletter
\providecommand \@ifxundefined [1]{%
 \@ifx{#1\undefined}
}%
\providecommand \@ifnum [1]{%
 \ifnum #1\expandafter \@firstoftwo
 \else \expandafter \@secondoftwo
 \fi
}%
\providecommand \@ifx [1]{%
 \ifx #1\expandafter \@firstoftwo
 \else \expandafter \@secondoftwo
 \fi
}%
\providecommand \natexlab [1]{#1}%
\providecommand \enquote  [1]{``#1''}%
\providecommand \bibnamefont  [1]{#1}%
\providecommand \bibfnamefont [1]{#1}%
\providecommand \citenamefont [1]{#1}%
\providecommand \href@noop [0]{\@secondoftwo}%
\providecommand \href [0]{\begingroup \@sanitize@url \@href}%
\providecommand \@href[1]{\@@startlink{#1}\@@href}%
\providecommand \@@href[1]{\endgroup#1\@@endlink}%
\providecommand \@sanitize@url [0]{\catcode `\\12\catcode `\$12\catcode `\&12\catcode `\#12\catcode `\^12\catcode `\_12\catcode `\%12\relax}%
\providecommand \@@startlink[1]{}%
\providecommand \@@endlink[0]{}%
\providecommand \url  [0]{\begingroup\@sanitize@url \@url }%
\providecommand \@url [1]{\endgroup\@href {#1}{\urlprefix }}%
\providecommand \urlprefix  [0]{URL }%
\providecommand \Eprint [0]{\href }%
\providecommand \doibase [0]{https://doi.org/}%
\providecommand \selectlanguage [0]{\@gobble}%
\providecommand \bibinfo  [0]{\@secondoftwo}%
\providecommand \bibfield  [0]{\@secondoftwo}%
\providecommand \translation [1]{[#1]}%
\providecommand \BibitemOpen [0]{}%
\providecommand \bibitemStop [0]{}%
\providecommand \bibitemNoStop [0]{.\EOS\space}%
\providecommand \EOS [0]{\spacefactor3000\relax}%
\providecommand \BibitemShut  [1]{\csname bibitem#1\endcsname}%
\let\auto@bib@innerbib\@empty
\bibitem [{\citenamefont {Lloyd}(1996)}]{Lloyd}%
  \BibitemOpen
  \bibfield  {author} {\bibinfo {author} {\bibfnamefont {S.}~\bibnamefont {Lloyd}},\ }\bibfield  {title} {\bibinfo {title} {\href{http://www.jstor.org/stable/2899535}{Universal Quantum Simulators}},\ }\href {http://www.jstor.org/stable/2899535} {\bibfield  {journal} {\bibinfo  {journal} {Science}\ }\textbf {\bibinfo {volume} {273}},\ \bibinfo {pages} {1073} (\bibinfo {year} {1996})}\BibitemShut {NoStop}%
\bibitem [{\citenamefont {Poulin}\ \emph {et~al.}(2015)\citenamefont {Poulin}, \citenamefont {Hastings}, \citenamefont {Wecker}, \citenamefont {Wiebe}, \citenamefont {Doberty},\ and\ \citenamefont {Troyer}}]{Poulin2015TheTS}%
  \BibitemOpen
  \bibfield  {author} {\bibinfo {author} {\bibfnamefont {D.}~\bibnamefont {Poulin}}, \bibinfo {author} {\bibfnamefont {M.}~\bibnamefont {Hastings}}, \bibinfo {author} {\bibfnamefont {D.}~\bibnamefont {Wecker}}, \bibinfo {author} {\bibfnamefont {N.}~\bibnamefont {Wiebe}}, \bibinfo {author} {\bibfnamefont {A.~C.}\ \bibnamefont {Doberty}},\ and\ \bibinfo {author} {\bibfnamefont {M.}~\bibnamefont {Troyer}},\ }\bibfield  {title} {\bibinfo {title} {\href{https://www.microsoft.com/en-us/research/publication/the-trotter-step-size-required-for-accurate-quantum-simulation-of-quantum-chemistry}{The Trotter step size required for accurate quantum simulation of quantum chemistry}},\ }\href@noop {} {\bibfield  {journal} {\bibinfo  {journal} {Quantum Inf. Comput.}\ }\textbf {\bibinfo {volume} {15}},\ \bibinfo {pages} {361} (\bibinfo {year} {2015})}\BibitemShut {NoStop}%
\bibitem [{\citenamefont {Heyl}\ \emph {et~al.}(2019)\citenamefont {Heyl}, \citenamefont {Hauke},\ and\ \citenamefont {Zoller}}]{Heyleaau8342}%
  \BibitemOpen
  \bibfield  {author} {\bibinfo {author} {\bibfnamefont {M.}~\bibnamefont {Heyl}}, \bibinfo {author} {\bibfnamefont {P.}~\bibnamefont {Hauke}},\ and\ \bibinfo {author} {\bibfnamefont {P.}~\bibnamefont {Zoller}},\ }\bibfield  {title} {\bibinfo {title} {\href{https://advances.sciencemag.org/content/5/4/eaau8342.full.pdf}{Quantum localization bounds Trotter errors in digital quantum simulation}},\ }\href {https://doi.org/10.1126/sciadv.aau8342} {\bibfield  {journal} {\bibinfo  {journal} {Science Advances}\ }\textbf {\bibinfo {volume} {5}},\ \bibinfo {pages} {eaau8342} (\bibinfo {year} {2019})}\BibitemShut {NoStop}%
\bibitem [{\citenamefont {Benedetti}\ \emph {et~al.}(2020)\citenamefont {Benedetti}, \citenamefont {Fiorentini},\ and\ \citenamefont {Lubasch}}]{benedetti2020hardwareefficient}%
  \BibitemOpen
  \bibfield  {author} {\bibinfo {author} {\bibfnamefont {M.}~\bibnamefont {Benedetti}}, \bibinfo {author} {\bibfnamefont {M.}~\bibnamefont {Fiorentini}},\ and\ \bibinfo {author} {\bibfnamefont {M.}~\bibnamefont {Lubasch}},\ }\href@noop {} {\bibinfo {title} {\href{https://arxiv.org/abs/2009.12361}{Hardware-efficient variational quantum algorithms for time evolution}}} (\bibinfo {year} {2020}),\ \Eprint {https://arxiv.org/abs/2009.12361} {arXiv:2009.12361 [quant-ph]} \BibitemShut {NoStop}%
\bibitem [{\citenamefont {Clinton}\ \emph {et~al.}(2024)\citenamefont {Clinton}, \citenamefont {Cubitt}, \citenamefont {Flynn}, \citenamefont {Gambetta}, \citenamefont {Klassen}, \citenamefont {Montanaro}, \citenamefont {Piddock}, \citenamefont {Santos},\ and\ \citenamefont {Sheridan}}]{Clinton_2024}%
  \BibitemOpen
  \bibfield  {author} {\bibinfo {author} {\bibfnamefont {L.}~\bibnamefont {Clinton}}, \bibinfo {author} {\bibfnamefont {T.}~\bibnamefont {Cubitt}}, \bibinfo {author} {\bibfnamefont {B.}~\bibnamefont {Flynn}}, \bibinfo {author} {\bibfnamefont {F.~M.}\ \bibnamefont {Gambetta}}, \bibinfo {author} {\bibfnamefont {J.}~\bibnamefont {Klassen}}, \bibinfo {author} {\bibfnamefont {A.}~\bibnamefont {Montanaro}}, \bibinfo {author} {\bibfnamefont {S.}~\bibnamefont {Piddock}}, \bibinfo {author} {\bibfnamefont {R.~A.}\ \bibnamefont {Santos}},\ and\ \bibinfo {author} {\bibfnamefont {E.}~\bibnamefont {Sheridan}},\ }\bibfield  {title} {\bibinfo {title} {Towards near-term quantum simulation of materials},\ }\href {http://dx.doi.org/10.1038/s41467-023-43479-6} {\bibfield  {journal} {\bibinfo  {journal} {Nature Communications}\ }\textbf {\bibinfo {volume} {15}},\ \bibinfo {pages} {211} (\bibinfo {year} {2024})}\BibitemShut {NoStop}%
\bibitem [{\citenamefont {Chowdhury}\ \emph {et~al.}(2021)\citenamefont {Chowdhury}, \citenamefont {Somma},\ and\ \citenamefont {Subaşı}}]{Chowdhury_2021}%
  \BibitemOpen
  \bibfield  {author} {\bibinfo {author} {\bibfnamefont {A.~N.}\ \bibnamefont {Chowdhury}}, \bibinfo {author} {\bibfnamefont {R.~D.}\ \bibnamefont {Somma}},\ and\ \bibinfo {author} {\bibfnamefont {Y.}~\bibnamefont {Subaşı}},\ }\bibfield  {title} {\bibinfo {title} {Computing partition functions in the one-clean-qubit model},\ }\href {http://dx.doi.org/10.1103/PhysRevA.103.032422} {\bibfield  {journal} {\bibinfo  {journal} {Physical Review A}\ }\textbf {\bibinfo {volume} {103}},\ \bibinfo {pages} {032422} (\bibinfo {year} {2021})}\BibitemShut {NoStop}%
\bibitem [{\citenamefont {Jackson}\ \emph {et~al.}(2023)\citenamefont {Jackson}, \citenamefont {Kapourniotis},\ and\ \citenamefont {Datta}}]{Jackson_2023}%
  \BibitemOpen
  \bibfield  {author} {\bibinfo {author} {\bibfnamefont {A.}~\bibnamefont {Jackson}}, \bibinfo {author} {\bibfnamefont {T.}~\bibnamefont {Kapourniotis}},\ and\ \bibinfo {author} {\bibfnamefont {A.}~\bibnamefont {Datta}},\ }\bibfield  {title} {\bibinfo {title} {Partition-function estimation: Quantum and quantum-inspired algorithms},\ }\href {http://dx.doi.org/10.1103/PhysRevA.107.012421} {\bibfield  {journal} {\bibinfo  {journal} {Physical Review A}\ }\textbf {\bibinfo {volume} {107}},\ \bibinfo {pages} {012421} (\bibinfo {year} {2023})}\BibitemShut {NoStop}%
\bibitem [{\citenamefont {Cornelissen}\ and\ \citenamefont {Hamoudi}()}]{doi:10.1137/1.9781611977554.ch46}%
  \BibitemOpen
  \bibfield  {author} {\bibinfo {author} {\bibfnamefont {A.}~\bibnamefont {Cornelissen}}\ and\ \bibinfo {author} {\bibfnamefont {Y.}~\bibnamefont {Hamoudi}},\ }\bibinfo {title} {A sublinear-time quantum algorithm for approximating partition functions},\ in\ \href {https://doi.org/10.1137/1.9781611977554.ch46} {\emph {\bibinfo {booktitle} {Proceedings of the 2023 Annual ACM-SIAM Symposium on Discrete Algorithms (SODA)}}},\ pp.\ \bibinfo {pages} {1245--1264}\BibitemShut {NoStop}%
\bibitem [{\citenamefont {Daley}\ \emph {et~al.}(2022)\citenamefont {Daley}, \citenamefont {Bloch}, \citenamefont {Kokail}, \citenamefont {Flannigan}, \citenamefont {Pearson}, \citenamefont {Troyer},\ and\ \citenamefont {Zoller}}]{Daley2022}%
  \BibitemOpen
  \bibfield  {author} {\bibinfo {author} {\bibfnamefont {A.~J.}\ \bibnamefont {Daley}}, \bibinfo {author} {\bibfnamefont {I.}~\bibnamefont {Bloch}}, \bibinfo {author} {\bibfnamefont {C.}~\bibnamefont {Kokail}}, \bibinfo {author} {\bibfnamefont {S.}~\bibnamefont {Flannigan}}, \bibinfo {author} {\bibfnamefont {N.}~\bibnamefont {Pearson}}, \bibinfo {author} {\bibfnamefont {M.}~\bibnamefont {Troyer}},\ and\ \bibinfo {author} {\bibfnamefont {P.}~\bibnamefont {Zoller}},\ }\bibfield  {title} {\bibinfo {title} {Practical quantum advantage in quantum simulation},\ }\href {https://doi.org/10.1038/s41586-022-04940-6} {\bibfield  {journal} {\bibinfo  {journal} {Nature}\ }\textbf {\bibinfo {volume} {607}},\ \bibinfo {pages} {667} (\bibinfo {year} {2022})}\BibitemShut {NoStop}%
\bibitem [{\citenamefont {Bairey}\ \emph {et~al.}(2019)\citenamefont {Bairey}, \citenamefont {Arad},\ and\ \citenamefont {Lindner}}]{PhysRevLett.122.020504}%
  \BibitemOpen
  \bibfield  {author} {\bibinfo {author} {\bibfnamefont {E.}~\bibnamefont {Bairey}}, \bibinfo {author} {\bibfnamefont {I.}~\bibnamefont {Arad}},\ and\ \bibinfo {author} {\bibfnamefont {N.~H.}\ \bibnamefont {Lindner}},\ }\bibfield  {title} {\bibinfo {title} {Learning a local hamiltonian from local measurements},\ }\href {https://doi.org/10.1103/PhysRevLett.122.020504} {\bibfield  {journal} {\bibinfo  {journal} {Phys. Rev. Lett.}\ }\textbf {\bibinfo {volume} {122}},\ \bibinfo {pages} {020504} (\bibinfo {year} {2019})}\BibitemShut {NoStop}%
\bibitem [{\citenamefont {Carrasco}\ \emph {et~al.}(2021)\citenamefont {Carrasco}, \citenamefont {Elben}, \citenamefont {Kokail}, \citenamefont {Kraus},\ and\ \citenamefont {Zoller}}]{PRXQuantum.2.010102}%
  \BibitemOpen
  \bibfield  {author} {\bibinfo {author} {\bibfnamefont {J.}~\bibnamefont {Carrasco}}, \bibinfo {author} {\bibfnamefont {A.}~\bibnamefont {Elben}}, \bibinfo {author} {\bibfnamefont {C.}~\bibnamefont {Kokail}}, \bibinfo {author} {\bibfnamefont {B.}~\bibnamefont {Kraus}},\ and\ \bibinfo {author} {\bibfnamefont {P.}~\bibnamefont {Zoller}},\ }\bibfield  {title} {\bibinfo {title} {Theoretical and experimental perspectives of quantum verification},\ }\href {https://doi.org/10.1103/PRXQuantum.2.010102} {\bibfield  {journal} {\bibinfo  {journal} {PRX Quantum}\ }\textbf {\bibinfo {volume} {2}},\ \bibinfo {pages} {010102} (\bibinfo {year} {2021})}\BibitemShut {NoStop}%
\bibitem [{\citenamefont {Shaffer}\ \emph {et~al.}(2021)\citenamefont {Shaffer}, \citenamefont {Megidish}, \citenamefont {Broz}, \citenamefont {Chen},\ and\ \citenamefont {H\"{a}ffner}}]{Shaffer2021}%
  \BibitemOpen
  \bibfield  {author} {\bibinfo {author} {\bibfnamefont {R.}~\bibnamefont {Shaffer}}, \bibinfo {author} {\bibfnamefont {E.}~\bibnamefont {Megidish}}, \bibinfo {author} {\bibfnamefont {J.}~\bibnamefont {Broz}}, \bibinfo {author} {\bibfnamefont {W.-T.}\ \bibnamefont {Chen}},\ and\ \bibinfo {author} {\bibfnamefont {H.}~\bibnamefont {H\"{a}ffner}},\ }\bibfield  {title} {\bibinfo {title} {Practical verification protocols for analog quantum simulators},\ }\href {https://doi.org/10.1038/s41534-021-00380-8} {\bibfield  {journal} {\bibinfo  {journal} {npj Quantum Information}\ }\textbf {\bibinfo {volume} {7}} (\bibinfo {year} {2021})}\BibitemShut {NoStop}%
\bibitem [{\citenamefont {Jackson}\ \emph {et~al.}(2024)\citenamefont {Jackson}, \citenamefont {Kapourniotis},\ and\ \citenamefont {Datta}}]{jackson2023accreditation}%
  \BibitemOpen
  \bibfield  {author} {\bibinfo {author} {\bibfnamefont {A.}~\bibnamefont {Jackson}}, \bibinfo {author} {\bibfnamefont {T.}~\bibnamefont {Kapourniotis}},\ and\ \bibinfo {author} {\bibfnamefont {A.}~\bibnamefont {Datta}},\ }\bibfield  {title} {\bibinfo {title} {Accreditation of analogue quantum simulators},\ }\href {https://doi.org/10.1073/pnas.2309627121} {\bibfield  {journal} {\bibinfo  {journal} {Proceedings of the National Academy of Sciences}\ }\textbf {\bibinfo {volume} {121}},\ \bibinfo {pages} {e2309627121} (\bibinfo {year} {2024})}\BibitemShut {NoStop}%
\bibitem [{\citenamefont {Cubitt}\ \emph {et~al.}(2018)\citenamefont {Cubitt}, \citenamefont {Montanaro},\ and\ \citenamefont {Piddock}}]{Cubitt_2018}%
  \BibitemOpen
  \bibfield  {author} {\bibinfo {author} {\bibfnamefont {T.~S.}\ \bibnamefont {Cubitt}}, \bibinfo {author} {\bibfnamefont {A.}~\bibnamefont {Montanaro}},\ and\ \bibinfo {author} {\bibfnamefont {S.}~\bibnamefont {Piddock}},\ }\bibfield  {title} {\bibinfo {title} {Universal quantum hamiltonians},\ }\href {https://doi.org/10.1073/pnas.1804949115} {\bibfield  {journal} {\bibinfo  {journal} {Proceedings of the National Academy of Sciences}\ }\textbf {\bibinfo {volume} {115}},\ \bibinfo {pages} {9497–9502} (\bibinfo {year} {2018})}\BibitemShut {NoStop}%
\bibitem [{\citenamefont {Zhou}\ and\ \citenamefont {Aharonov}(2021)}]{zhou2021strongly}%
  \BibitemOpen
  \bibfield  {author} {\bibinfo {author} {\bibfnamefont {L.}~\bibnamefont {Zhou}}\ and\ \bibinfo {author} {\bibfnamefont {D.}~\bibnamefont {Aharonov}},\ }\href@noop {} {\bibinfo {title} {Strongly universal hamiltonian simulators}} (\bibinfo {year} {2021}),\ \Eprint {https://arxiv.org/abs/2102.02991} {arXiv:2102.02991 [quant-ph]} \BibitemShut {NoStop}%
\bibitem [{\citenamefont {Kohler}\ \emph {et~al.}(2022)\citenamefont {Kohler}, \citenamefont {Piddock}, \citenamefont {Bausch},\ and\ \citenamefont {Cubitt}}]{PRXQuantum.3.010308}%
  \BibitemOpen
  \bibfield  {author} {\bibinfo {author} {\bibfnamefont {T.}~\bibnamefont {Kohler}}, \bibinfo {author} {\bibfnamefont {S.}~\bibnamefont {Piddock}}, \bibinfo {author} {\bibfnamefont {J.}~\bibnamefont {Bausch}},\ and\ \bibinfo {author} {\bibfnamefont {T.}~\bibnamefont {Cubitt}},\ }\bibfield  {title} {\bibinfo {title} {General conditions for universality of quantum hamiltonians},\ }\href {https://doi.org/10.1103/PRXQuantum.3.010308} {\bibfield  {journal} {\bibinfo  {journal} {PRX Quantum}\ }\textbf {\bibinfo {volume} {3}},\ \bibinfo {pages} {010308} (\bibinfo {year} {2022})}\BibitemShut {NoStop}%
\bibitem [{\citenamefont {Wallman}\ and\ \citenamefont {Emerson}(2016)}]{Wallman_2016}%
  \BibitemOpen
  \bibfield  {author} {\bibinfo {author} {\bibfnamefont {J.~J.}\ \bibnamefont {Wallman}}\ and\ \bibinfo {author} {\bibfnamefont {J.}~\bibnamefont {Emerson}},\ }\bibfield  {title} {\bibinfo {title} {Noise tailoring for scalable quantum computation via randomized compiling},\ }\href {https://doi.org/10.1103%2Fphysreva.94.052325} {\bibfield  {journal} {\bibinfo  {journal} {Physical Review A}\ }\textbf {\bibinfo {volume} {94}},\ \bibinfo {pages} {052325} (\bibinfo {year} {2016})}\BibitemShut {NoStop}%
\bibitem [{\citenamefont {Hashim}\ \emph {et~al.}(2021)\citenamefont {Hashim}, \citenamefont {Naik}, \citenamefont {Morvan}, \citenamefont {Ville}, \citenamefont {Mitchell}, \citenamefont {Kreikebaum}, \citenamefont {Davis}, \citenamefont {Smith}, \citenamefont {Iancu}, \citenamefont {O’Brien}, \citenamefont {Hincks}, \citenamefont {Wallman}, \citenamefont {Emerson}, \citenamefont {Siddiqi} \emph {et~al.}}]{Hashim_2021}%
  \BibitemOpen
  \bibfield  {author} {\bibinfo {author} {\bibfnamefont {A.}~\bibnamefont {Hashim}}, \bibinfo {author} {\bibfnamefont {R.~K.}\ \bibnamefont {Naik}}, \bibinfo {author} {\bibfnamefont {A.}~\bibnamefont {Morvan}}, \bibinfo {author} {\bibfnamefont {J.-L.}\ \bibnamefont {Ville}}, \bibinfo {author} {\bibfnamefont {B.}~\bibnamefont {Mitchell}}, \bibinfo {author} {\bibfnamefont {J.~M.}\ \bibnamefont {Kreikebaum}}, \bibinfo {author} {\bibfnamefont {M.}~\bibnamefont {Davis}}, \bibinfo {author} {\bibfnamefont {E.}~\bibnamefont {Smith}}, \bibinfo {author} {\bibfnamefont {C.}~\bibnamefont {Iancu}}, \bibinfo {author} {\bibfnamefont {K.~P.}\ \bibnamefont {O’Brien}}, \bibinfo {author} {\bibfnamefont {I.}~\bibnamefont {Hincks}}, \bibinfo {author} {\bibfnamefont {J.~J.}\ \bibnamefont {Wallman}}, \bibinfo {author} {\bibfnamefont {J.}~\bibnamefont {Emerson}}, \bibinfo {author} {\bibfnamefont {I.}~\bibnamefont {Siddiqi}}, \emph {et~al.},\ }\bibfield  {title} {\bibinfo {title} {Randomized compiling for scalable quantum computing
  on a noisy superconducting quantum processor},\ }\href {http://dx.doi.org/10.1103/PhysRevX.11.041039} {\bibfield  {journal} {\bibinfo  {journal} {Physical Review X}\ }\textbf {\bibinfo {volume} {11}},\ \bibinfo {pages} {041039} (\bibinfo {year} {2021})}\BibitemShut {NoStop}%
\bibitem [{\citenamefont {Feng}\ \emph {et~al.}(2023)\citenamefont {Feng}, \citenamefont {Liu},\ and\ \citenamefont {Liu}}]{feng2023deterministically}%
  \BibitemOpen
  \bibfield  {author} {\bibinfo {author} {\bibfnamefont {W.}~\bibnamefont {Feng}}, \bibinfo {author} {\bibfnamefont {L.}~\bibnamefont {Liu}},\ and\ \bibinfo {author} {\bibfnamefont {T.}~\bibnamefont {Liu}},\ }\href@noop {} {\bibinfo {title} {On deterministically approximating total variation distance}} (\bibinfo {year} {2023}),\ \Eprint {https://arxiv.org/abs/2309.14696} {arXiv:2309.14696 [cs.DS]} \BibitemShut {NoStop}%
\bibitem [{\citenamefont {Gao}\ \emph {et~al.}(2017)\citenamefont {Gao}, \citenamefont {Wang},\ and\ \citenamefont {Duan}}]{PhysRevLett.118.040502}%
  \BibitemOpen
  \bibfield  {author} {\bibinfo {author} {\bibfnamefont {X.}~\bibnamefont {Gao}}, \bibinfo {author} {\bibfnamefont {S.-T.}\ \bibnamefont {Wang}},\ and\ \bibinfo {author} {\bibfnamefont {L.-M.}\ \bibnamefont {Duan}},\ }\bibfield  {title} {\bibinfo {title} {Quantum supremacy for simulating a translation-invariant ising spin model},\ }\href {https://doi.org/10.1103/PhysRevLett.118.040502} {\bibfield  {journal} {\bibinfo  {journal} {Phys. Rev. Lett.}\ }\textbf {\bibinfo {volume} {118}},\ \bibinfo {pages} {040502} (\bibinfo {year} {2017})}\BibitemShut {NoStop}%
\bibitem [{\citenamefont {Bermejo-Vega}\ \emph {et~al.}(2018)\citenamefont {Bermejo-Vega}, \citenamefont {Hangleiter}, \citenamefont {Schwarz}, \citenamefont {Raussendorf},\ and\ \citenamefont {Eisert}}]{PhysRevX.8.021010}%
  \BibitemOpen
  \bibfield  {author} {\bibinfo {author} {\bibfnamefont {J.}~\bibnamefont {Bermejo-Vega}}, \bibinfo {author} {\bibfnamefont {D.}~\bibnamefont {Hangleiter}}, \bibinfo {author} {\bibfnamefont {M.}~\bibnamefont {Schwarz}}, \bibinfo {author} {\bibfnamefont {R.}~\bibnamefont {Raussendorf}},\ and\ \bibinfo {author} {\bibfnamefont {J.}~\bibnamefont {Eisert}},\ }\bibfield  {title} {\bibinfo {title} {Architectures for quantum simulation showing a quantum speedup},\ }\href {https://doi.org/10.1103/PhysRevX.8.021010} {\bibfield  {journal} {\bibinfo  {journal} {Phys. Rev. X}\ }\textbf {\bibinfo {volume} {8}},\ \bibinfo {pages} {021010} (\bibinfo {year} {2018})}\BibitemShut {NoStop}%
\bibitem [{\citenamefont {Kapourniotis}\ and\ \citenamefont {Datta}(2019)}]{Kapourniotis_2019}%
  \BibitemOpen
  \bibfield  {author} {\bibinfo {author} {\bibfnamefont {T.}~\bibnamefont {Kapourniotis}}\ and\ \bibinfo {author} {\bibfnamefont {A.}~\bibnamefont {Datta}},\ }\bibfield  {title} {\bibinfo {title} {Nonadaptive fault-tolerant verification of quantum supremacy with noise},\ }\href {https://doi.org/10.22331/q-2019-07-12-164} {\bibfield  {journal} {\bibinfo  {journal} {Quantum}\ }\textbf {\bibinfo {volume} {3}},\ \bibinfo {pages} {164} (\bibinfo {year} {2019})}\BibitemShut {NoStop}%
\bibitem [{\citenamefont {Ringbauer}\ \emph {et~al.}(2024)\citenamefont {Ringbauer}, \citenamefont {Hinsche}, \citenamefont {Feldker}, \citenamefont {Faehrmann}, \citenamefont {Bermejo-Vega}, \citenamefont {Edmunds}, \citenamefont {Postler}, \citenamefont {Stricker}, \citenamefont {Marciniak}, \citenamefont {Meth}, \citenamefont {Pogorelov}, \citenamefont {Blatt}, \citenamefont {Schindler}, \citenamefont {Eisert}, \citenamefont {Monz},\ and\ \citenamefont {Hangleiter}}]{ringbauer2024verifiablemeasurementbasedquantumrandom}%
  \BibitemOpen
  \bibfield  {author} {\bibinfo {author} {\bibfnamefont {M.}~\bibnamefont {Ringbauer}}, \bibinfo {author} {\bibfnamefont {M.}~\bibnamefont {Hinsche}}, \bibinfo {author} {\bibfnamefont {T.}~\bibnamefont {Feldker}}, \bibinfo {author} {\bibfnamefont {P.~K.}\ \bibnamefont {Faehrmann}}, \bibinfo {author} {\bibfnamefont {J.}~\bibnamefont {Bermejo-Vega}}, \bibinfo {author} {\bibfnamefont {C.}~\bibnamefont {Edmunds}}, \bibinfo {author} {\bibfnamefont {L.}~\bibnamefont {Postler}}, \bibinfo {author} {\bibfnamefont {R.}~\bibnamefont {Stricker}}, \bibinfo {author} {\bibfnamefont {C.~D.}\ \bibnamefont {Marciniak}}, \bibinfo {author} {\bibfnamefont {M.}~\bibnamefont {Meth}}, \bibinfo {author} {\bibfnamefont {I.}~\bibnamefont {Pogorelov}}, \bibinfo {author} {\bibfnamefont {R.}~\bibnamefont {Blatt}}, \bibinfo {author} {\bibfnamefont {P.}~\bibnamefont {Schindler}}, \bibinfo {author} {\bibfnamefont {J.}~\bibnamefont {Eisert}}, \bibinfo {author} {\bibfnamefont {T.}~\bibnamefont {Monz}},\ and\ \bibinfo {author} {\bibfnamefont
  {D.}~\bibnamefont {Hangleiter}},\ }\href {https://arxiv.org/abs/2307.14424} {\bibinfo {title} {Verifiable measurement-based quantum random sampling with trapped ions}} (\bibinfo {year} {2024}),\ \Eprint {https://arxiv.org/abs/2307.14424} {arXiv:2307.14424 [quant-ph]} \BibitemShut {NoStop}%
\bibitem [{\citenamefont {Liu}\ \emph {et~al.}(2024)\citenamefont {Liu}, \citenamefont {Devulapalli}, \citenamefont {Hangleiter}, \citenamefont {Liu}, \citenamefont {Kollár}, \citenamefont {Gorshkov},\ and\ \citenamefont {Childs}}]{liu2024efficientlyverifiablequantumadvantage}%
  \BibitemOpen
  \bibfield  {author} {\bibinfo {author} {\bibfnamefont {Z.}~\bibnamefont {Liu}}, \bibinfo {author} {\bibfnamefont {D.}~\bibnamefont {Devulapalli}}, \bibinfo {author} {\bibfnamefont {D.}~\bibnamefont {Hangleiter}}, \bibinfo {author} {\bibfnamefont {Y.-K.}\ \bibnamefont {Liu}}, \bibinfo {author} {\bibfnamefont {A.~J.}\ \bibnamefont {Kollár}}, \bibinfo {author} {\bibfnamefont {A.~V.}\ \bibnamefont {Gorshkov}},\ and\ \bibinfo {author} {\bibfnamefont {A.~M.}\ \bibnamefont {Childs}},\ }\href {https://arxiv.org/abs/2403.08195} {\bibinfo {title} {Efficiently verifiable quantum advantage on near-term analog quantum simulators}} (\bibinfo {year} {2024}),\ \Eprint {https://arxiv.org/abs/2403.08195} {arXiv:2403.08195 [quant-ph]} \BibitemShut {NoStop}%
\bibitem [{\citenamefont {Odake}\ \emph {et~al.}(2024)\citenamefont {Odake}, \citenamefont {Kristjánsson}, \citenamefont {Taranto},\ and\ \citenamefont {Murao}}]{odake2024universal}%
  \BibitemOpen
  \bibfield  {author} {\bibinfo {author} {\bibfnamefont {T.}~\bibnamefont {Odake}}, \bibinfo {author} {\bibfnamefont {H.}~\bibnamefont {Kristjánsson}}, \bibinfo {author} {\bibfnamefont {P.}~\bibnamefont {Taranto}},\ and\ \bibinfo {author} {\bibfnamefont {M.}~\bibnamefont {Murao}},\ }\href@noop {} {\bibinfo {title} {Universal algorithm for transforming hamiltonian eigenvalues}} (\bibinfo {year} {2024}),\ \Eprint {https://arxiv.org/abs/2312.08848} {arXiv:2312.08848 [quant-ph]} \BibitemShut {NoStop}%
\bibitem [{\citenamefont {Jackson}(2024)}]{jackson2024accreditationlimitedadversarialnoise}%
  \BibitemOpen
  \bibfield  {author} {\bibinfo {author} {\bibfnamefont {A.}~\bibnamefont {Jackson}},\ }\href {https://arxiv.org/abs/2409.03995} {\bibinfo {title} {Accreditation against limited adversarial noise}} (\bibinfo {year} {2024}),\ \Eprint {https://arxiv.org/abs/2409.03995} {arXiv:2409.03995 [quant-ph]} \BibitemShut {NoStop}%
\bibitem [{\citenamefont {Derbyshire}\ \emph {et~al.}(2020)\citenamefont {Derbyshire}, \citenamefont {Malo}, \citenamefont {Daley}, \citenamefont {Kashefi},\ and\ \citenamefont {Wallden}}]{2020AnalogueBenchmarking}%
  \BibitemOpen
  \bibfield  {author} {\bibinfo {author} {\bibfnamefont {E.}~\bibnamefont {Derbyshire}}, \bibinfo {author} {\bibfnamefont {J.~Y.}\ \bibnamefont {Malo}}, \bibinfo {author} {\bibfnamefont {A.~J.}\ \bibnamefont {Daley}}, \bibinfo {author} {\bibfnamefont {E.}~\bibnamefont {Kashefi}},\ and\ \bibinfo {author} {\bibfnamefont {P.}~\bibnamefont {Wallden}},\ }\bibfield  {title} {\bibinfo {title} {Randomized benchmarking in the analogue setting},\ }\href {https://doi.org/10.1088/2058-9565/ab7eec} {\bibfield  {journal} {\bibinfo  {journal} {Quantum Science and Technology}\ }\textbf {\bibinfo {volume} {5}},\ \bibinfo {pages} {034001} (\bibinfo {year} {2020})}\BibitemShut {NoStop}%
\bibitem [{\citenamefont {Ferracin}\ \emph {et~al.}(2021{\natexlab{a}})\citenamefont {Ferracin}, \citenamefont {Merkel}, \citenamefont {McKay},\ and\ \citenamefont {Datta}}]{2021Sams}%
  \BibitemOpen
  \bibfield  {author} {\bibinfo {author} {\bibfnamefont {S.}~\bibnamefont {Ferracin}}, \bibinfo {author} {\bibfnamefont {S.~T.}\ \bibnamefont {Merkel}}, \bibinfo {author} {\bibfnamefont {D.}~\bibnamefont {McKay}},\ and\ \bibinfo {author} {\bibfnamefont {A.}~\bibnamefont {Datta}},\ }\bibfield  {title} {\bibinfo {title} {Experimental accreditation of outputs of noisy quantum computers},\ }\href {https://link.aps.org/doi/10.1103/PhysRevA.104.042603} {\bibfield  {journal} {\bibinfo  {journal} {Phys. Rev. A}\ }\textbf {\bibinfo {volume} {104}},\ \bibinfo {pages} {042603} (\bibinfo {year} {2021}{\natexlab{a}})}\BibitemShut {NoStop}%
\bibitem [{\citenamefont {Erhard}\ \emph {et~al.}(2019)\citenamefont {Erhard}, \citenamefont {Wallman}, \citenamefont {Postler}, \citenamefont {Meth}, \citenamefont {Stricker}, \citenamefont {Martinez}, \citenamefont {Schindler}, \citenamefont {Monz}, \citenamefont {Emerson},\ and\ \citenamefont {Blatt}}]{Erhard2019}%
  \BibitemOpen
  \bibfield  {author} {\bibinfo {author} {\bibfnamefont {A.}~\bibnamefont {Erhard}}, \bibinfo {author} {\bibfnamefont {J.~J.}\ \bibnamefont {Wallman}}, \bibinfo {author} {\bibfnamefont {L.}~\bibnamefont {Postler}}, \bibinfo {author} {\bibfnamefont {M.}~\bibnamefont {Meth}}, \bibinfo {author} {\bibfnamefont {R.}~\bibnamefont {Stricker}}, \bibinfo {author} {\bibfnamefont {E.~A.}\ \bibnamefont {Martinez}}, \bibinfo {author} {\bibfnamefont {P.}~\bibnamefont {Schindler}}, \bibinfo {author} {\bibfnamefont {T.}~\bibnamefont {Monz}}, \bibinfo {author} {\bibfnamefont {J.}~\bibnamefont {Emerson}},\ and\ \bibinfo {author} {\bibfnamefont {R.}~\bibnamefont {Blatt}},\ }\bibfield  {title} {\bibinfo {title} {Characterizing large-scale quantum computers via cycle benchmarking},\ }\href {https://doi.org/10.1038/s41467-019-13068-7} {\bibfield  {journal} {\bibinfo  {journal} {Nature Communications}\ }\textbf {\bibinfo {volume} {10}},\ \bibinfo {pages} {5347} (\bibinfo {year} {2019})}\BibitemShut {NoStop}%
\bibitem [{\citenamefont {Carignan-Dugas}\ \emph {et~al.}(2015)\citenamefont {Carignan-Dugas}, \citenamefont {Wallman},\ and\ \citenamefont {Emerson}}]{PhysRevA.92.060302}%
  \BibitemOpen
  \bibfield  {author} {\bibinfo {author} {\bibfnamefont {A.}~\bibnamefont {Carignan-Dugas}}, \bibinfo {author} {\bibfnamefont {J.~J.}\ \bibnamefont {Wallman}},\ and\ \bibinfo {author} {\bibfnamefont {J.}~\bibnamefont {Emerson}},\ }\bibfield  {title} {\bibinfo {title} {Characterizing universal gate sets via dihedral benchmarking},\ }\href {https://doi.org/10.1103/PhysRevA.92.060302} {\bibfield  {journal} {\bibinfo  {journal} {Phys. Rev. A}\ }\textbf {\bibinfo {volume} {92}},\ \bibinfo {pages} {060302} (\bibinfo {year} {2015})}\BibitemShut {NoStop}%
\bibitem [{\citenamefont {Lu}\ \emph {et~al.}(2015)\citenamefont {Lu}, \citenamefont {Li}, \citenamefont {Trottier}, \citenamefont {Li}, \citenamefont {Brodutch}, \citenamefont {Krismanich}, \citenamefont {Ghavami}, \citenamefont {Dmitrienko}, \citenamefont {Long}, \citenamefont {Baugh},\ and\ \citenamefont {Laflamme}}]{PhysRevLett.114.140505}%
  \BibitemOpen
  \bibfield  {author} {\bibinfo {author} {\bibfnamefont {D.}~\bibnamefont {Lu}}, \bibinfo {author} {\bibfnamefont {H.}~\bibnamefont {Li}}, \bibinfo {author} {\bibfnamefont {D.-A.}\ \bibnamefont {Trottier}}, \bibinfo {author} {\bibfnamefont {J.}~\bibnamefont {Li}}, \bibinfo {author} {\bibfnamefont {A.}~\bibnamefont {Brodutch}}, \bibinfo {author} {\bibfnamefont {A.~P.}\ \bibnamefont {Krismanich}}, \bibinfo {author} {\bibfnamefont {A.}~\bibnamefont {Ghavami}}, \bibinfo {author} {\bibfnamefont {G.~I.}\ \bibnamefont {Dmitrienko}}, \bibinfo {author} {\bibfnamefont {G.}~\bibnamefont {Long}}, \bibinfo {author} {\bibfnamefont {J.}~\bibnamefont {Baugh}},\ and\ \bibinfo {author} {\bibfnamefont {R.}~\bibnamefont {Laflamme}},\ }\bibfield  {title} {\bibinfo {title} {Experimental estimation of average fidelity of a clifford gate on a 7-qubit quantum processor},\ }\href {https://doi.org/10.1103/PhysRevLett.114.140505} {\bibfield  {journal} {\bibinfo  {journal} {Phys. Rev. Lett.}\ }\textbf {\bibinfo {volume} {114}},\ \bibinfo
  {pages} {140505} (\bibinfo {year} {2015})}\BibitemShut {NoStop}%
\bibitem [{\citenamefont {Merkel}\ \emph {et~al.}(2013)\citenamefont {Merkel}, \citenamefont {Gambetta}, \citenamefont {Smolin}, \citenamefont {Poletto}, \citenamefont {C\'orcoles}, \citenamefont {Johnson}, \citenamefont {Ryan},\ and\ \citenamefont {Steffen}}]{PhysRevA.87.062119}%
  \BibitemOpen
  \bibfield  {author} {\bibinfo {author} {\bibfnamefont {S.~T.}\ \bibnamefont {Merkel}}, \bibinfo {author} {\bibfnamefont {J.~M.}\ \bibnamefont {Gambetta}}, \bibinfo {author} {\bibfnamefont {J.~A.}\ \bibnamefont {Smolin}}, \bibinfo {author} {\bibfnamefont {S.}~\bibnamefont {Poletto}}, \bibinfo {author} {\bibfnamefont {A.~D.}\ \bibnamefont {C\'orcoles}}, \bibinfo {author} {\bibfnamefont {B.~R.}\ \bibnamefont {Johnson}}, \bibinfo {author} {\bibfnamefont {C.~A.}\ \bibnamefont {Ryan}},\ and\ \bibinfo {author} {\bibfnamefont {M.}~\bibnamefont {Steffen}},\ }\bibfield  {title} {\bibinfo {title} {Self-consistent quantum process tomography},\ }\href {https://doi.org/10.1103/PhysRevA.87.062119} {\bibfield  {journal} {\bibinfo  {journal} {Phys. Rev. A}\ }\textbf {\bibinfo {volume} {87}},\ \bibinfo {pages} {062119} (\bibinfo {year} {2013})}\BibitemShut {NoStop}%
\bibitem [{\citenamefont {Flammia}\ and\ \citenamefont {Liu}(2011)}]{PhysRevLett.106.230501}%
  \BibitemOpen
  \bibfield  {author} {\bibinfo {author} {\bibfnamefont {S.~T.}\ \bibnamefont {Flammia}}\ and\ \bibinfo {author} {\bibfnamefont {Y.-K.}\ \bibnamefont {Liu}},\ }\bibfield  {title} {\bibinfo {title} {Direct fidelity estimation from few pauli measurements},\ }\href {https://doi.org/10.1103/PhysRevLett.106.230501} {\bibfield  {journal} {\bibinfo  {journal} {Phys. Rev. Lett.}\ }\textbf {\bibinfo {volume} {106}},\ \bibinfo {pages} {230501} (\bibinfo {year} {2011})}\BibitemShut {NoStop}%
\bibitem [{\citenamefont {Dahlhauser}\ and\ \citenamefont {Humble}(2021)}]{Dahlhauser_2021}%
  \BibitemOpen
  \bibfield  {author} {\bibinfo {author} {\bibfnamefont {M.~L.}\ \bibnamefont {Dahlhauser}}\ and\ \bibinfo {author} {\bibfnamefont {T.~S.}\ \bibnamefont {Humble}},\ }\bibfield  {title} {\bibinfo {title} {Modeling noisy quantum circuits using experimental characterization},\ }\href {https://doi.org/10.1103%2Fphysreva.103.042603} {\bibfield  {journal} {\bibinfo  {journal} {Physical Review A}\ }\textbf {\bibinfo {volume} {103}},\ \bibinfo {pages} {042603} (\bibinfo {year} {2021})}\BibitemShut {NoStop}%
\bibitem [{\citenamefont {Moussa}\ \emph {et~al.}(2012)\citenamefont {Moussa}, \citenamefont {da~Silva}, \citenamefont {Ryan},\ and\ \citenamefont {Laflamme}}]{PhysRevLett.109.070504}%
  \BibitemOpen
  \bibfield  {author} {\bibinfo {author} {\bibfnamefont {O.}~\bibnamefont {Moussa}}, \bibinfo {author} {\bibfnamefont {M.~P.}\ \bibnamefont {da~Silva}}, \bibinfo {author} {\bibfnamefont {C.~A.}\ \bibnamefont {Ryan}},\ and\ \bibinfo {author} {\bibfnamefont {R.}~\bibnamefont {Laflamme}},\ }\bibfield  {title} {\bibinfo {title} {Practical experimental certification of computational quantum gates using a twirling procedure},\ }\href {https://doi.org/10.1103/PhysRevLett.109.070504} {\bibfield  {journal} {\bibinfo  {journal} {Phys. Rev. Lett.}\ }\textbf {\bibinfo {volume} {109}},\ \bibinfo {pages} {070504} (\bibinfo {year} {2012})}\BibitemShut {NoStop}%
\bibitem [{\citenamefont {Harper}\ \emph {et~al.}(2020)\citenamefont {Harper}, \citenamefont {Flammia},\ and\ \citenamefont {Wallman}}]{Harper2020}%
  \BibitemOpen
  \bibfield  {author} {\bibinfo {author} {\bibfnamefont {R.}~\bibnamefont {Harper}}, \bibinfo {author} {\bibfnamefont {S.~T.}\ \bibnamefont {Flammia}},\ and\ \bibinfo {author} {\bibfnamefont {J.~J.}\ \bibnamefont {Wallman}},\ }\bibfield  {title} {\bibinfo {title} {Efficient learning of quantum noise},\ }\href {https://doi.org/10.1038/s41567-020-0992-8} {\bibfield  {journal} {\bibinfo  {journal} {Nature Physics}\ }\textbf {\bibinfo {volume} {16}},\ \bibinfo {pages} {1184} (\bibinfo {year} {2020})}\BibitemShut {NoStop}%
\bibitem [{\citenamefont {Ferracin}\ \emph {et~al.}(2022)\citenamefont {Ferracin}, \citenamefont {Hashim}, \citenamefont {Ville}, \citenamefont {Naik}, \citenamefont {Carignan-Dugas}, \citenamefont {Qassim}, \citenamefont {Morvan}, \citenamefont {Santiago}, \citenamefont {Siddiqi},\ and\ \citenamefont {Wallman}}]{ferracin2022efficiently}%
  \BibitemOpen
  \bibfield  {author} {\bibinfo {author} {\bibfnamefont {S.}~\bibnamefont {Ferracin}}, \bibinfo {author} {\bibfnamefont {A.}~\bibnamefont {Hashim}}, \bibinfo {author} {\bibfnamefont {J.-L.}\ \bibnamefont {Ville}}, \bibinfo {author} {\bibfnamefont {R.}~\bibnamefont {Naik}}, \bibinfo {author} {\bibfnamefont {A.}~\bibnamefont {Carignan-Dugas}}, \bibinfo {author} {\bibfnamefont {H.}~\bibnamefont {Qassim}}, \bibinfo {author} {\bibfnamefont {A.}~\bibnamefont {Morvan}}, \bibinfo {author} {\bibfnamefont {D.~I.}\ \bibnamefont {Santiago}}, \bibinfo {author} {\bibfnamefont {I.}~\bibnamefont {Siddiqi}},\ and\ \bibinfo {author} {\bibfnamefont {J.~J.}\ \bibnamefont {Wallman}},\ }\href@noop {} {\bibinfo {title} {Efficiently improving the performance of noisy quantum computers}} (\bibinfo {year} {2022}),\ \Eprint {https://arxiv.org/abs/2201.10672} {arXiv:2201.10672 [quant-ph]} \BibitemShut {NoStop}%
\bibitem [{\citenamefont {Emerson}\ \emph {et~al.}(2007)\citenamefont {Emerson}, \citenamefont {Silva}, \citenamefont {Moussa}, \citenamefont {Ryan}, \citenamefont {Laforest}, \citenamefont {Baugh}, \citenamefont {Cory},\ and\ \citenamefont {Laflamme}}]{doi:10.1126/science.1145699}%
  \BibitemOpen
  \bibfield  {author} {\bibinfo {author} {\bibfnamefont {J.}~\bibnamefont {Emerson}}, \bibinfo {author} {\bibfnamefont {M.}~\bibnamefont {Silva}}, \bibinfo {author} {\bibfnamefont {O.}~\bibnamefont {Moussa}}, \bibinfo {author} {\bibfnamefont {C.}~\bibnamefont {Ryan}}, \bibinfo {author} {\bibfnamefont {M.}~\bibnamefont {Laforest}}, \bibinfo {author} {\bibfnamefont {J.}~\bibnamefont {Baugh}}, \bibinfo {author} {\bibfnamefont {D.~G.}\ \bibnamefont {Cory}},\ and\ \bibinfo {author} {\bibfnamefont {R.}~\bibnamefont {Laflamme}},\ }\bibfield  {title} {\bibinfo {title} {Symmetrized characterization of noisy quantum processes},\ }\href {https://doi.org/10.1126/science.1145699} {\bibfield  {journal} {\bibinfo  {journal} {Science}\ }\textbf {\bibinfo {volume} {317}},\ \bibinfo {pages} {1893} (\bibinfo {year} {2007})}\BibitemShut {NoStop}%
\bibitem [{\citenamefont {Harty}\ \emph {et~al.}(2014)\citenamefont {Harty}, \citenamefont {Allcock}, \citenamefont {Ballance}, \citenamefont {Guidoni}, \citenamefont {Janacek}, \citenamefont {Linke}, \citenamefont {Stacey},\ and\ \citenamefont {Lucas}}]{PhysRevLett.113.220501}%
  \BibitemOpen
  \bibfield  {author} {\bibinfo {author} {\bibfnamefont {T.~P.}\ \bibnamefont {Harty}}, \bibinfo {author} {\bibfnamefont {D.~T.~C.}\ \bibnamefont {Allcock}}, \bibinfo {author} {\bibfnamefont {C.~J.}\ \bibnamefont {Ballance}}, \bibinfo {author} {\bibfnamefont {L.}~\bibnamefont {Guidoni}}, \bibinfo {author} {\bibfnamefont {H.~A.}\ \bibnamefont {Janacek}}, \bibinfo {author} {\bibfnamefont {N.~M.}\ \bibnamefont {Linke}}, \bibinfo {author} {\bibfnamefont {D.~N.}\ \bibnamefont {Stacey}},\ and\ \bibinfo {author} {\bibfnamefont {D.~M.}\ \bibnamefont {Lucas}},\ }\bibfield  {title} {\bibinfo {title} {High-fidelity preparation, gates, memory, and readout of a trapped-ion quantum bit},\ }\href {https://doi.org/10.1103/PhysRevLett.113.220501} {\bibfield  {journal} {\bibinfo  {journal} {Phys. Rev. Lett.}\ }\textbf {\bibinfo {volume} {113}},\ \bibinfo {pages} {220501} (\bibinfo {year} {2014})}\BibitemShut {NoStop}%
\bibitem [{\citenamefont {Wright}\ \emph {et~al.}(2019)\citenamefont {Wright}, \citenamefont {Beck}, \citenamefont {Debnath}, \citenamefont {Amini}, \citenamefont {Nam}, \citenamefont {Grzesiak} \emph {et~al.}}]{Wright2019}%
  \BibitemOpen
  \bibfield  {author} {\bibinfo {author} {\bibfnamefont {K.}~\bibnamefont {Wright}}, \bibinfo {author} {\bibfnamefont {K.~M.}\ \bibnamefont {Beck}}, \bibinfo {author} {\bibfnamefont {S.}~\bibnamefont {Debnath}}, \bibinfo {author} {\bibfnamefont {J.~M.}\ \bibnamefont {Amini}}, \bibinfo {author} {\bibfnamefont {Y.}~\bibnamefont {Nam}}, \bibinfo {author} {\bibfnamefont {N.}~\bibnamefont {Grzesiak}}, \emph {et~al.},\ }\bibfield  {title} {\bibinfo {title} {Benchmarking an 11-qubit quantum computer},\ }\href {https://doi.org/10.1038/s41467-019-13534-2} {\bibfield  {journal} {\bibinfo  {journal} {Nature Communications}\ }\textbf {\bibinfo {volume} {10}},\ \bibinfo {pages} {5464} (\bibinfo {year} {2019})}\BibitemShut {NoStop}%
\bibitem [{\citenamefont {Arute}\ \emph {et~al.}(2019)\citenamefont {Arute}, \citenamefont {Arya}, \citenamefont {Babbush}, \citenamefont {Bacon}, \citenamefont {Bardin} \emph {et~al.}}]{Arute2019}%
  \BibitemOpen
  \bibfield  {author} {\bibinfo {author} {\bibfnamefont {F.}~\bibnamefont {Arute}}, \bibinfo {author} {\bibfnamefont {K.}~\bibnamefont {Arya}}, \bibinfo {author} {\bibfnamefont {R.}~\bibnamefont {Babbush}}, \bibinfo {author} {\bibfnamefont {D.}~\bibnamefont {Bacon}}, \bibinfo {author} {\bibfnamefont {J.~C.}\ \bibnamefont {Bardin}}, \emph {et~al.},\ }\bibfield  {title} {\bibinfo {title} {Quantum supremacy using a programmable superconducting processor},\ }\href {https://doi.org/10.1038/s41586-019-1666-5} {\bibfield  {journal} {\bibinfo  {journal} {Nature}\ }\textbf {\bibinfo {volume} {574}},\ \bibinfo {pages} {505} (\bibinfo {year} {2019})}\BibitemShut {NoStop}%
\bibitem [{\citenamefont {Ferracin}\ \emph {et~al.}(2019)\citenamefont {Ferracin}, \citenamefont {Kapourniotis},\ and\ \citenamefont {Datta}}]{Ferracin_2019}%
  \BibitemOpen
  \bibfield  {author} {\bibinfo {author} {\bibfnamefont {S.}~\bibnamefont {Ferracin}}, \bibinfo {author} {\bibfnamefont {T.}~\bibnamefont {Kapourniotis}},\ and\ \bibinfo {author} {\bibfnamefont {A.}~\bibnamefont {Datta}},\ }\bibfield  {title} {\bibinfo {title} {Accrediting outputs of noisy intermediate-scale quantum computing devices},\ }\href {https://doi.org/10.1088/1367-2630/ab4fd6} {\bibfield  {journal} {\bibinfo  {journal} {New Journal of Physics}\ }\textbf {\bibinfo {volume} {21}},\ \bibinfo {pages} {113038} (\bibinfo {year} {2019})}\BibitemShut {NoStop}%
\bibitem [{\citenamefont {Ferracin}\ \emph {et~al.}(2021{\natexlab{b}})\citenamefont {Ferracin}, \citenamefont {Merkel}, \citenamefont {McKay},\ and\ \citenamefont {Datta}}]{Ferracin_2021}%
  \BibitemOpen
  \bibfield  {author} {\bibinfo {author} {\bibfnamefont {S.}~\bibnamefont {Ferracin}}, \bibinfo {author} {\bibfnamefont {S.~T.}\ \bibnamefont {Merkel}}, \bibinfo {author} {\bibfnamefont {D.}~\bibnamefont {McKay}},\ and\ \bibinfo {author} {\bibfnamefont {A.}~\bibnamefont {Datta}},\ }\bibfield  {title} {\bibinfo {title} {Experimental accreditation of outputs of noisy quantum computers},\ }\href {https://doi.org/10.1103%2Fphysreva.104.042603} {\bibfield  {journal} {\bibinfo  {journal} {Physical Review A}\ }\textbf {\bibinfo {volume} {104}},\ \bibinfo {pages} {042603} (\bibinfo {year} {2021}{\natexlab{b}})}\BibitemShut {NoStop}%
\bibitem [{\citenamefont {Hoeffding}(1963)}]{doi:10.1080/01621459.1963.10500830}%
  \BibitemOpen
  \bibfield  {author} {\bibinfo {author} {\bibfnamefont {W.}~\bibnamefont {Hoeffding}},\ }\bibfield  {title} {\bibinfo {title} {Probability inequalities for sums of bounded random variables},\ }\href {https://doi.org/10.1080/01621459.1963.10500830} {\bibfield  {journal} {\bibinfo  {journal} {Journal of the American Statistical Association}\ }\textbf {\bibinfo {volume} {58}},\ \bibinfo {pages} {13} (\bibinfo {year} {1963})}\BibitemShut {NoStop}%
\bibitem [{\citenamefont {Preskill}(2012)}]{Preskill:2012tg}%
  \BibitemOpen
  \bibfield  {author} {\bibinfo {author} {\bibfnamefont {J.}~\bibnamefont {Preskill}},\ }\href {https://arxiv.org/abs/1203.5813} {\bibinfo {title} {{Quantum computing and the entanglement frontier}}} (\bibinfo {year} {2012}),\ \Eprint {https://arxiv.org/abs/1203.5813} {arXiv:1203.5813 [quant-ph]} \BibitemShut {NoStop}%
\bibitem [{\citenamefont {Harrow}\ and\ \citenamefont {Montanaro}(2017)}]{Harrow_2017}%
  \BibitemOpen
  \bibfield  {author} {\bibinfo {author} {\bibfnamefont {A.~W.}\ \bibnamefont {Harrow}}\ and\ \bibinfo {author} {\bibfnamefont {A.}~\bibnamefont {Montanaro}},\ }\bibfield  {title} {\bibinfo {title} {Quantum computational supremacy},\ }\href {https://doi.org/10.1038/nature23458} {\bibfield  {journal} {\bibinfo  {journal} {Nature}\ }\textbf {\bibinfo {volume} {549}},\ \bibinfo {pages} {203–209} (\bibinfo {year} {2017})}\BibitemShut {NoStop}%
\bibitem [{\citenamefont {Fujii}\ and\ \citenamefont {Morimae}(2017)}]{Fujii_2017}%
  \BibitemOpen
  \bibfield  {author} {\bibinfo {author} {\bibfnamefont {K.}~\bibnamefont {Fujii}}\ and\ \bibinfo {author} {\bibfnamefont {T.}~\bibnamefont {Morimae}},\ }\bibfield  {title} {\bibinfo {title} {Commuting quantum circuits and complexity of ising partition functions},\ }\href {https://doi.org/10.1088/1367-2630/aa5fdb} {\bibfield  {journal} {\bibinfo  {journal} {New Journal of Physics}\ }\textbf {\bibinfo {volume} {19}},\ \bibinfo {pages} {033003} (\bibinfo {year} {2017})}\BibitemShut {NoStop}%
\bibitem [{\citenamefont {Browaeys}\ and\ \citenamefont {Lahaye}(2020)}]{Browaeys2020}%
  \BibitemOpen
  \bibfield  {author} {\bibinfo {author} {\bibfnamefont {A.}~\bibnamefont {Browaeys}}\ and\ \bibinfo {author} {\bibfnamefont {T.}~\bibnamefont {Lahaye}},\ }\bibfield  {title} {\bibinfo {title} {Many-body physics with individually controlled rydberg atoms},\ }\href {https://doi.org/10.1038/s41567-019-0733-z} {\bibfield  {journal} {\bibinfo  {journal} {Nature Physics}\ }\textbf {\bibinfo {volume} {16}},\ \bibinfo {pages} {132} (\bibinfo {year} {2020})}\BibitemShut {NoStop}%
\bibitem [{\citenamefont {Bluvstein}\ \emph {et~al.}(2022)\citenamefont {Bluvstein}, \citenamefont {Levine}, \citenamefont {Semeghini}, \citenamefont {Wang}, \citenamefont {Ebadi}, \citenamefont {Kalinowski}, \citenamefont {Keesling}, \citenamefont {Maskara}, \citenamefont {Pichler}, \citenamefont {Greiner}, \citenamefont {Vuleti{\'{c}}},\ and\ \citenamefont {Lukin}}]{Bluvstein2022}%
  \BibitemOpen
  \bibfield  {author} {\bibinfo {author} {\bibfnamefont {D.}~\bibnamefont {Bluvstein}}, \bibinfo {author} {\bibfnamefont {H.}~\bibnamefont {Levine}}, \bibinfo {author} {\bibfnamefont {G.}~\bibnamefont {Semeghini}}, \bibinfo {author} {\bibfnamefont {T.~T.}\ \bibnamefont {Wang}}, \bibinfo {author} {\bibfnamefont {S.}~\bibnamefont {Ebadi}}, \bibinfo {author} {\bibfnamefont {M.}~\bibnamefont {Kalinowski}}, \bibinfo {author} {\bibfnamefont {A.}~\bibnamefont {Keesling}}, \bibinfo {author} {\bibfnamefont {N.}~\bibnamefont {Maskara}}, \bibinfo {author} {\bibfnamefont {H.}~\bibnamefont {Pichler}}, \bibinfo {author} {\bibfnamefont {M.}~\bibnamefont {Greiner}}, \bibinfo {author} {\bibfnamefont {V.}~\bibnamefont {Vuleti{\'{c}}}},\ and\ \bibinfo {author} {\bibfnamefont {M.~D.}\ \bibnamefont {Lukin}},\ }\bibfield  {title} {\bibinfo {title} {A quantum processor based on coherent transport of entangled atom arrays},\ }\href {https://doi.org/10.1038/s41586-022-04592-6} {\bibfield  {journal} {\bibinfo  {journal} {Nature}\
  }\textbf {\bibinfo {volume} {604}},\ \bibinfo {pages} {451} (\bibinfo {year} {2022})}\BibitemShut {NoStop}%
\end{thebibliography}%


\onecolumngrid
\appendix
\section{Proof of Theorem~\ref{thm:inversioWorks}}
\label{sec:AppendixProofOfTheorem:inversionWorks}
\subsection{Basic Definitions in Graph Theory}
To enable the discussion in this appendix, we first present foundational definitions of graphs and hypergraphs.
\begin{definition}
    \label{def:hypergraph}
    A \underline{hypergraph}, $h = (V, E)$, is a set of two sets: $V$ and $E$.
    Elements of $V$ are referred to as vertices and each element in $E$ (referred to as a hyperedge) is a subset of $V$. If each element in $E$ is a double of elements in $V$, they are referred to as edges and the hypergraph is a \underline{graph}. 
\end{definition}
An important concept of this appendix is colouring a graph.
\begin{definition}
\label{def:vertexColour}
\label{def:Colouring}
    For any hypergraph, $h = (\mathbf{V}_{\textit{h}}, \mathbf{E}_{\textit{h}})$, a vertex is \underline{coloured} when assigned an element of $\mathbb{N}$. The element of $\mathbb{N}$ a vertex is assigned is referred to as its \underline{colour}.
    A \underline{colouring} of a hypergraph is a mapping, $\chi$: $V \longrightarrow$ $\{1, 2, ..., k \in \mathbb{N}\}$, such that:
    \begin{align}
        &\forall e \in \mathbf{E}_{\textit{h}},
        \forall v_1 \in e,  v_2 \not = v_1\in e,
        \chi(v_1) \neq  \chi(v_2).
    \end{align}
\end{definition}
\begin{definition}
\label{kColDef}
    A hypergraph, $\textit{h} = (\mathbf{V}_{\textit{h}}, \mathbf{E}_{\textit{h}})$, is \underline{k-colourable} if and only if there exists a colouring mapping, $\chi$: $\mathbf{V}_{\textit{h}} \longrightarrow$ $\{1, 2, ..., k\}$, i.e. a colouring (as in Def.~\ref{def:Colouring}) with the image $\{1, 2, ..., k\}$ so it assigns each vertex of $\textit{h}$ one of $k$ colours.
\end{definition}
\begin{definition}
\label{cromatSetDef}
Define the \underline{chromatic subsets} for a hypergraph $h$, corresponding to a given colouring, $\chi$, by:
\begin{align}
    \Tilde{\mathbb{V}}_j &= \big \{ v \in \mathbf{V}_{\textit{h}} \text{ }\vert \text{ } \chi(v) = j \big \},
\end{align}
where the underlying hypergraph, $\textit{h}$, has been k-coloured for some integer $k \geq j$. The \underline{chromatic number} of a hypergraph is the number of chromatic subsets in the colouring of the hypergraph with the minimum number of chromatic subsets. For a hypergraph, $\textit{h}$, we denote this by $C_{\textit{rom}} \big( \textit{h} \big)$.
\end{definition}
\begin{definition}
    The \underline{maximum vertex degree} of a hypergraph is the maximum number number of hyperedges including a specific vertex in that graph. For a hypergraph, $\textit{h}$, we denote this by $\Delta \big( \textit{h} \big)$.
\end{definition}

\subsection{Interaction Hypergraph Preparatory Definition and Results}
\begin{definition}
    We define the \underline{greedy colouring} of a graph as the colouring of that graph generated by Algorithm~\ref{alg:greedyColour}
    \begin{figure}
    \centering
\begin{algorithm}[H]
$\mathbf{Input:}$ \\
$\bullet$ A graph, $\mathbf{G} = (\mathbf{V}, \mathbf{E})$\\
\vspace{0.3cm}
\For{$v$ in $\mathbf{V}$}{
    \begin{enumerate}
        \item unUsed = $\mathbb{N}$
        \item \For{every vertex, $k$, sharing a hyperedge with $v$}{
        \begin{enumerate}
            \item \If {$k$ is already coloured}{
                \begin{enumerate}
                    \item Remove the colour of $k$ from unUsed 
                \end{enumerate}
            }
        \end{enumerate}
        }
        \item Assign to $v$ the minimum value in unUsed as its colour
    \end{enumerate}
    }
\vspace{0.1cm}
$\mathbf{Return}:$ The colouring constructed above
\caption{Greedy algorithm for colouring a graph
 \label{alg:greedyColour}}
\end{algorithm}
\end{figure}
\end{definition}
We then begin the journey towards proving Theorem~\ref{thm:inversioWorks} with some relevant graph theory to represent the interactions in the Hamiltonian being considered. This notion is formalized in Def.~\ref{def:hypergraphRepresent}.
\begin{definition}
\label{def:hypergraphRepresent}
For any Hamiltonian acting on a system of $N$ spins, interactions between the spins that occur in the Hamiltonian, can be represented by a hypergraph, $\big ( \mathbf{V}, \mathbf{E} \big)$, known as the \underline{hypergraph representing the Hamiltonian}. $\mathbf{V}$ is a set of each of the indices labelling a spin. $\mathbf{E}$ defines the edges in the hypergraph and is a set that has a single element for each term in the Hamiltonian. For each term in the Hamiltonian, the corresponding element of $\mathbf{E}$ is a set of the indices of all the spins which that term acts non-trivially on.
\end{definition}
It is useful to establish that finite hypergraphs are always finitely colourable, in Lemma~\ref{lem:colourHypergraph}.
\begin{lemma}
    \label{lem:colourHypergraph}
    For any finite hypergraph,  $\textit{h} = (\mathbf{V}_{\textit{h}}, \mathbf{E}_{\textit{h}})$, there exists some finite $k \in \mathbb{N}$ such that $h$ is $k$-colourable.
\end{lemma}
\begin{proof}
    Define an auxillary graph $\mathbf{G}_{\textit{h}}$ with the exact same set of vertices as $\textit{h}$ but define it's set of edges, $\mathbf{E}_G$, by:
    for every hyperedge, $e$, in $\textit{h}$,
    \begin{align}
        \forall j \in e, \forall k \not = j \in e, (j,k) \in \mathbf{E}_G.
    \end{align}
    Assume we have an invalid colouring of $\textit{h}$ that assigns a colour to every vertex, labeled $\mathcal{F}$. As it is invalid, $\exists e \in \mathbf{E}_{\textit{h}}$ such that $\exists j, k \in e$ such that $\mathcal{F} (j) = \mathcal{F} (k)$. But as $\exists e \in \mathbf{E}_{\textit{h}}$ such that $j, k \in e$, by the definition of $\mathbf{G}_{\textit{h}}$, $(j,k) \in \mathbf{E}_G$. As $\mathcal{F} (j) = \mathcal{F} (k)$ and $(j,k) \in \mathbf{E}_G$, $\mathcal{F}$ is not a valid colouring of $\mathbf{G}_{\textit{h}}$. Hence if the colouring is invalid for $\textit{h}$, it is invalid for $\mathbf{G}_{\textit{h}}$. Taking the contrapositive of this, any valid colouring of $\mathbf{G}_{\textit{h}}$ is a valid colouring of $\textit{h}$. Colourings of finite graphs are always finite. Hence for some $k \in \mathbb{N}$,  $\textit{h}$ is $k$-colourable.
\end{proof}
The existence of a finite colouring of any finite hypergraph, as established in Lemma.~\ref{lem:colourHypergraph}, is used in Lemma~\ref{lem:GExists}.
\begin{lemma}
\label{lem:GExists}
    For any spin Hamiltonian, $\mathcal{H} = \sum_{q = 1} \big( \mathcal{H}_q \big)$, the corresponding $\mathbb{G}_{\mathcal{H}}$ exists.
\end{lemma}
\begin{proof}
    Given $\mathcal{H}$, the corresponding hypergraph (as defined in Def.~\ref{def:hypergraphRepresent}), $h$, can be defined. By Lemma~\ref{lem:colourHypergraph}, for any $\mathcal{H}$, $h$ is $k$-colourable (for some finite $k \in \mathbb{N}$) so let $\nu_1$, $\nu_2$, ..., $\nu_k$ be the chromatic sets of one such $k$-colouring of $h$.
    Let $\mathcal{\sigma}^{(\nu_j)}_{v_i}$ denote a set of Pauli gates -- specified by the string $v_j \in \{0, 1, 2, 3 \}^{\vert \nu_j \vert}$ -- acting on the qubits with indices in the chromatic set, $\nu_j$. Then we define:
    \begin{align}
        \mathbb{G}_{\mathcal{H}} = \bigcup^k_{j=1} \bigg( \bigg\{ \mathcal{\sigma}^{(\nu_j)}_{v_j} \text{ } \big \vert \text{ }v_j \in \{0, 1, 2, 3 \}^{\vert \nu_j \vert} \bigg \} \bigg).
    \end{align}
    Any term $\mathcal{H}_q$ in $\mathcal{H}$, corresponds to a single hyperedge, $e$, in the corresponding hypergraph, $h$. By the definition of a colouring of a hypergraph, there must exists at least one chromatic set, $\nu_i$, such that $\vert e \cap \nu_i \vert = 1$.  So for any chromatic set, $v_j$, $\mathcal{\sigma}^{(\nu_j)}_{v_j}$ and $\mathcal{H}_q$ both act non-trivially on only a single qubit (i.e. there is only a single qubit that is acted on by both $\mathcal{\sigma}^{(\nu_j)}_{v_j}$ and $\mathcal{H}_q$). Therefore, a choice of $v_i$ can be made so that $\mathcal{\sigma}^{(\nu_i)}_{v_i}$ and $\mathcal{H}_q$ anti-commute.
\end{proof}
\subsection{Proving Theorem~\ref{thm:inversioWorks}: Main Section}
The diamond norm has many useful properties, but the most useful for this paper is the chain property, proven in Lemma~\ref{ChainingPropertyLemma}.
\begin{lemma}
    \label{ChainingPropertyLemma}
    Any sub-multiplicative norm bounded by one on CPTP maps, $\big \vert \big \vert \cdot  \big \vert \big \vert_{\mathrm{sub}}$, exhibits the chaining property, i.e. for any operators, $\mathcal{A}, \mathcal{B}, \mathcal{C}, \mathcal{D}$:
    \begin{align}
        \big \vert \big \vert \mathcal{A}\mathcal{B} - \mathcal{C}\mathcal{D} \big \vert \big \vert_{\mathrm{sub}} 
        &\leq
        \big \vert \big \vert \mathcal{A} - \mathcal{C} \big \vert \big \vert_{\mathrm{sub}} + \big \vert \big \vert \mathcal{B} - \mathcal{D} \big \vert \big \vert_{\mathrm{sub}}.
    \end{align}
    Examples of norms this applies to include the diamond norm and the trace norm (nuclear norm).
\end{lemma}
\begin{proof}
        \begin{align}
        \big \vert \big \vert \mathcal{A}\mathcal{B} - \mathcal{C}\mathcal{D} \big \vert \big \vert_{\mathrm{sub}}
        =
        \big \vert \big \vert \mathcal{A}\mathcal{B} -\mathcal{C}\mathcal{D} + \mathcal{A}\mathcal{D} - \mathcal{A}\mathcal{D} \big \vert \big \vert_{\mathrm{sub}}
        \leq
        \big \vert \big \vert \mathcal{A} \vert \big \vert_{\mathrm{sub}} \cdot \big \vert \big \vert \mathcal{B} - \mathcal{D} \big \vert \big \vert_{\mathrm{sub}} + \big \vert \big \vert \mathcal{A} - \mathcal{C} \big \vert \big \vert_{\mathrm{sub}} \cdot \big \vert \big \vert \mathcal{D} \big \vert \big \vert_{\mathrm{sub}}
        \leq
        \big \vert \big \vert \mathcal{A} - \mathcal{C} \big \vert \big \vert_{\mathrm{sub}}
        +
        \big \vert \big \vert \mathcal{B} - \mathcal{D} \big \vert \big \vert_{\mathrm{sub}}.
        \end{align}
\end{proof}
I can now start on the direct path towards proving Theorem~\ref{thm:inversioWorks}. The first step is Lemma~\ref{lem:sumOverSigmas}, which is purely mathematical.
\begin{lemma}[Lemma 9, Ref.~\cite{odake2024universal}]
    \label{lem:sumOverSigmas}
    $\forall J \subseteq \{ 0, 1, 2, 3 \}^N$, for any Hamiltonian, $\mathcal{H} \in \textit{span} \big( \{ \mathcal{P}_{\mathbf{u}} \}_{\mathbf{u} \in J} \big)$,
    \begin{align}
        \sum_{\mathcal{\sigma} \in \mathbb{G}'_{\mathcal{H}}} \bigg( \mathcal{\sigma} \mathcal{H} \mathcal{\sigma} \bigg)
        &=
        \dfrac{L \textit{Tr} \big( \mathcal{H} \big)}{2^N} \mathcal{I},
    \end{align}
    where $\mathbb{G}'_{\mathcal{H}}$ is as in Def.~\ref{def:G'}, $L$ is as in Def.~\ref{def:L}, and $N$ is the number of qubits $\mathcal{H}$ acts on.  
\end{lemma}
We are then fully equipped to finish the proof of  Theorem~\ref{thm:inversioWorks}.
\begin{proof}[Proof of Theorem~\ref{thm:inversioWorks}]
    Let $\Phi$ be a CPTP map denoting the implementation of the process defined in Algorithm~\ref{alg:addingSinglesForInvert} to the argument density matrix, $\rho$. Then, letting $\mathbb{G}_{\mathcal{H}}'' = \mathbb{G}_{\mathcal{H}}' \backslash \big \{ \mathcal{I} \big \}$,
    \begin{align}
    \label{eqn:ERhoInitial}
        \Phi \big( \rho \big)
        &=
        \dfrac{1}{L - 1} \sum_{\mathcal{\sigma} \in \mathbb{G}_{\mathcal{H}}''} \bigg( 
            \mathcal{\sigma} e^{-i t \mathcal{H} / M} \mathcal{\sigma} \rho \mathcal{\sigma} e^{i t \mathcal{H} / M} \mathcal{\sigma}
        \bigg).
    \end{align}
    Taking the Taylor expansion of the exponentials in the summand of Eqn.~\ref{eqn:ERhoInitial}:
    \begin{align}
        \label{eqn:CPTPwithExponentialsExpanded}
        \Phi \big( \rho \big)
        &=
        \dfrac{1}{L - 1} \sum_{\mathcal{\sigma} \in \mathbb{G}_{\mathcal{H}}''} \bigg( 
            \mathcal{\sigma} \bigg [ \mathcal{I} - \dfrac{i t \mathcal{H}}{M}  \bigg] \mathcal{\sigma} \rho \mathcal{\sigma} \bigg [ \mathcal{I} + \dfrac{i t \mathcal{H}}{M}  \bigg] \mathcal{\sigma}
        \bigg)
        +
        O \bigg( \dfrac{t^2}{M^2} \bigg)\\
        &=
        \dfrac{1}{L - 1} \sum_{\mathcal{\sigma} \in \mathbb{G}_{\mathcal{H}}''} \bigg(
            \rho
            \bigg)
            +
            \dfrac{i t}{(L - 1)M} \bigg[
            \rho \sum_{\mathcal{\sigma} \in \mathbb{G}_{\mathcal{H}}''} \bigg( \mathcal{\sigma} \mathcal{H} \mathcal{\sigma}
        \bigg)
        - \sum_{\mathcal{\sigma} \in \mathbb{G}_{\mathcal{H}}''} \bigg(
            \mathcal{\sigma}\mathcal{H} \mathcal{\sigma} \bigg)\rho
        \bigg]
        +
        O \bigg( \dfrac{t^2}{M^2} \bigg).
    \end{align}
    Then, using Lemma~\ref{lem:sumOverSigmas} to evaluate the summations in Eqn.~\ref{eqn:CPTPwithExponentialsExpanded}:
    \begin{align}
        \Phi \big( \rho \big)
        &=
        \rho
            +
        \dfrac{i t}{(L - 1)M} \bigg[
        \rho \bigg\{ - \mathcal{H} + \dfrac{L \textit{Tr} \big( \mathcal{H}\big)}{2^N} \mathcal{I} \bigg\}
        - \bigg\{ - \mathcal{H} + \dfrac{L \textit{Tr} \big( \mathcal{H}\big)}{2^N} \mathcal{I} \bigg\} \rho
        \bigg]
        +
        O \bigg( \dfrac{t^2}{M^2} \bigg)
        =
        \rho
        +
        \dfrac{i t}{(L - 1)M} \bigg[
          \mathcal{H} \rho - \rho \mathcal{H}
        \bigg]
        +
        O \bigg( \dfrac{t^2}{M^2} \bigg).
    \end{align}
    Consider a correctly implemented reverse time evolution and take the Taylor series of the exponentials:
    \begin{align}
        &e^{i \mathcal{H}t/M} \rho e^{-i \mathcal{H}t/M}
        =
        \big[ \mathcal{I} + \dfrac{i \mathcal{H}t}{M} \big] \rho \big[ \mathcal{I} -  \dfrac{i \mathcal{H}t}{M} \big]  +
        O \bigg( \dfrac{t^2}{M^2} \bigg)
        =
        \rho
        +
         \dfrac{i \mathcal{H}t}{M}  \rho 
        -
        \rho  \dfrac{i \mathcal{H}t}{M}
         +
        O \bigg( \dfrac{t^2}{M^2} \bigg)
        =
        \rho
        +
         \dfrac{i t}{M} \bigg[ \mathcal{H} \rho 
        -
        \rho \mathcal{H} \bigg]
         +
        O \bigg( \dfrac{t^2}{M^2} \bigg)\\
        \label{eqn:RoughBound}
        \Rightarrow &\bigg \vert \bigg \vert \Phi \big( \rho \big) - e^{i \mathcal{H}t/(M[L -1])} \rho e^{-i \mathcal{H}t/(M[L -1])} \bigg \vert \bigg \vert_{\diamond}
        =
         \bigg \vert \bigg \vert \text{ } \rho
        +
        \dfrac{i t}{(L - 1)M} \bigg[
          \mathcal{H} \rho - \rho \mathcal{H}
        \bigg]
        -
        \rho
        -
        \dfrac{i t}{(L - 1)M} \bigg[
          \mathcal{H} \rho - \rho \mathcal{H}
        \bigg]
         \bigg \vert \bigg \vert_{\diamond} +
        O \bigg( \dfrac{t^2}{M^2} \bigg)
        =
        O \bigg( \dfrac{t^2}{M^2} \bigg).
    \end{align}
    If $M \in \mathbb{N}$ iterations of this time evolution are applied\footnote{Where $\Phi^M$ is a CPTP map denoting $\Phi$ being applied $M \in \mathbb{N}$ times.}. Using the chaining property (from Lemma~\ref{ChainingPropertyLemma}) of the diamond norm:
    \begin{align}
        \bigg \vert \bigg \vert \Phi^M \big( \rho \big) - e^{i \mathcal{H}t/(L -1)} \rho e^{-i \mathcal{H}t/(L - 1)} \bigg \vert \bigg \vert_{\diamond}
        &\leq
        M \bigg \vert \bigg \vert \Phi \big( \rho \big) - e^{i \mathcal{H}t/(M[L -1])} \rho e^{-i \mathcal{H}t/(M[L -1])} \bigg \vert \bigg \vert_{\diamond}
        =
        O \bigg( \dfrac{t^2}{M} \bigg).
    \end{align}
\end{proof}
To examine the exact error in $\Phi^M \big( \rho \big)$, consider just the higher order terms of the Taylor expansions in  Eqn.~\ref{eqn:ERhoInitial}:
\begin{align}
          \bigg \vert \bigg \vert \Phi \big( \rho \big) - e^{i \mathcal{H}t/(M[L -1])} \rho e^{-i \mathcal{H}t/(M[L -1])} \bigg \vert \bigg \vert_{\diamond}
          &=
          \bigg \vert \bigg \vert
            \sum_{n + m \geq 2} \bigg( \dfrac{1}{n! m!} \dfrac{1}{L - 1} \sum_{\mathcal{\sigma} \in \mathbb{G}_{\mathcal{H}}''} \bigg(\mathcal{\sigma} \bigg[ \dfrac{-it \mathcal{H}}{M } \bigg]^n \mathcal{\sigma} \rho \mathcal{\sigma} \bigg[ \dfrac{it \mathcal{H}}{M } \bigg]^m \mathcal{\sigma} \bigg)
            -
            \dfrac{1}{n! m!}
            \bigg[ \dfrac{it \mathcal{H}}{M (L - 1)} \bigg]^n \rho \bigg[ \dfrac{-it \mathcal{H}}{M (L - 1)} \bigg]^m\bigg)
         \bigg \vert \bigg \vert_{\diamond}\\
         &\leq
            \dfrac{1}{n! m! (L - 1)} \sum_{n + m \geq 2} \bigg( \sum_{\mathcal{\sigma} \in \mathbb{G}_{\mathcal{H}}''} \bigg( \bigg[ \dfrac{t \big \vert \big \vert \mathcal{H} \big \vert \big \vert_{\diamond}}{M} \bigg]^{n+m}
            \bigg)
            +
            \bigg[ \dfrac{t \big \vert \big \vert \mathcal{H} \big \vert \big \vert_{\diamond}}{M} \bigg]^{n+m}
         \bigg)
        = \dfrac{L}{ (L - 1)} \bigg( e^{2 \frac{t  \vert \vert \mathcal{H} \vert \vert_{\diamond}}{M}} - 2 \dfrac{t \big \vert \big \vert \mathcal{H} \big \vert \big \vert_{\diamond}}{M} - 1 \bigg).
    \end{align}
If we want to bound the error, $\big \vert \big \vert \Phi^M \big( \rho \big) - e^{i \mathcal{H}t/(L -1)} \rho e^{-i \mathcal{H}t/(L - 1)} \big \vert \big \vert_{\diamond}$, by $\epsilon \in \mathbb{R}$ we set $M$ such that:
\begin{align}
        M \geq
        \dfrac{2 t^2 \big \vert \big \vert \mathcal{H} \big \vert \big \vert^2_{\diamond} L}{\epsilon (L - 1)}
        \Rightarrow
        \bigg \vert \bigg \vert \Phi^M \big( \rho \big) - e^{i \mathcal{H}t/(L -1)} \rho e^{-i \mathcal{H}t/(L - 1)} \bigg \vert \bigg \vert_{\diamond}
        \leq
        \dfrac{L}{ (L - 1)} \bigg( e^{ 2t  \vert \vert \mathcal{H} \vert \vert_{\diamond}/M} - 2 \dfrac{t \big \vert \big \vert \mathcal{H} \big \vert \big \vert_{\diamond}}{M} - 1 \bigg)
         \leq
         \epsilon,
\end{align}
where we have assumed that $t \big \vert \big \vert \mathcal{H} \big \vert \big \vert_{\diamond}$ is small enough to neglect powers of it higher than two.
$\big \vert \big \vert \mathcal{H} \big \vert \big \vert_{\diamond}$ can be  upper bounded, assuming that $\mathcal{H} = \sum_{j = 1} \big( c_j \mathcal{H}_j \big)$ (where each $\mathcal{H}_j$ is a Pauli string), by:
    $\sum_{j = 1} \bigg( \vert c_j \vert \bigg)$. This is exactly the value of $M$ used in Algorithm~\ref{alg:addingSinglesForInvert}, proving its correctness.
\end{document}